\documentclass[12pt]{article}
\usepackage[english]{babel}
\usepackage{amsthm,amssymb,amsfonts,amsmath,amscd}
\usepackage{graphicx}
\usepackage{nccmath}

\makeatletter \@addtoreset{equation}{section} \makeatother
\newtheorem{proposition}{Proposition}
\newtheorem{theorem}{Theorem}
\newtheorem{lemma}{Lemma}
\newtheorem{remark}{Remark}

\def\mdet{\mathrm{det}}
\def\intd{\displaystyle\int}
\def\Tr{\mathrm{Tr}\,}

\begin{document}

\title{On the second mixed moment of the characteristic polynomials of 1D band matrices}

\author{ Tatyana Shcherbina\thanks {Supported by NSF grant DMS 1128155}
\\ \\
\small{\it IAS, Princeton, USA}
 \\\small{\it e-mail: t\underline{ }shcherbina@rambler.ru}
}
\date{}
\maketitle

\begin{abstract}
We consider the asymptotic behavior of the second mixed moment of the characteristic polynomials of 1D
Gaussian band matrices, i.e. of the Hermitian $N\times N$ matrices $H_N$ with independent Gaussian entries such that
 $\langle H_{ij}H_{lk}\rangle=\delta_{ik}\delta_{jl}J_{ij}$, where $J=(-W^2\triangle+1)^{-1}$.
Assuming that $W^2=N^{1+\theta}$, $0<\theta\le 1$, we show that the moment's asymptotic behavior (as $N\to\infty$) in
the bulk of the spectrum coincides with that for the Gaussian Unitary Ensemble.
\end{abstract}
\section{Introduction}
The Hermitian Gaussian random band matrices (RBM) are Hermitian $N\times N$
matrices $H_N$ (we enumerate indices of entries by $i,j\in \mathcal{L}$, where
$\mathcal{L}=[-n,n]^d\cap \mathbb{Z}^d$, $N=(2n+1)^d$) whose entries $H_{ij}$ are random
Gaussian variables with mean zero such that
\begin{equation}\label{ban}
\mathbf{E}\big\{ H_{ij}H_{lk}\big\}=\delta_{ik}\delta_{jl}J_{ij},
\end{equation}
where $J_{ij}$ is a symmetric function which is small for large $|i-j|$ and
\[
\sum\limits_{i=-n}^nJ_{ij}=1.
\]
In this paper we consider the especially convenient choice of $J_{ij}$, which is given by the lattice Green's
function
\begin{equation}\label{J}
J_{ij}=\left(-W^2\Delta+1\right)^{-1}_{ij},
\end{equation}
where $\Delta$ is the discrete Laplacian on $\mathcal{L}$. For the case $d=1$
\begin{equation}\label{lapl}
(-\Delta f)_j=\left\{\begin{array}{ll} -f_{j-1}+2f_j-f_{j+1},&j\ne -n,n,\\
-f_{j-1}+f_j-f_{j+1},& j=-n,n
\end{array}\right.
\end{equation}
with $f_{-n-1}=f_{n+1}=0$ (i.e. we consider the discrete Laplacian with Neumann boundary conditions).
The advantage of (\ref{J}) is that its inverse, which appears
in the integral representation (see (\ref{HS}) below), is only three-diagonal matrix.

Note that $J_{ij}\approx C_1W^{-1}\exp\{-C_2|i-j|/W\}$ for $J$ of (\ref{J}) with $d=1$, and so the variance of matrix
elements is exponentially small when $|i-j|\gg W$. Hence $W$ can be considered as the width of the band.

It should be noted also that the odd size of the matrices is chosen only because it is
more convenient to have the symmetric segment $[-n, n]$ and it does not play any role
in the consideration below.

The probability law of 1D RBM $H_N$ can be written in the form
\begin{equation}\label{band}
P_n(d H_N)=\prod\limits_{-n\le i<j\le n}\dfrac{dH_{ij}d\overline{H}_{ij}}{2\pi J_{ij}}e^{-\frac{|H_{ij}|^2}{J_{ij}}}
\prod\limits_{i=-n}^n\dfrac{dH_{ii}}{\sqrt{2\pi J_{ii}}}e^{-\frac{H_{ii}^2}{2J_{ii}}}.
\end{equation}

Varying $W$, we can see that random band matrices are natural interpolations between random
Schr$\ddot{\hbox{o}}$dinger matrices $ H_{RS}=-\Delta+\lambda V$, in which the randomness
only appears in the diagonal potential $V$ ($\lambda$ is a small parameter which
measures the strength of the disorder) and
mean-field random matrices such as $N\times N$ Wigner matrices, i.e. Hermitian
random matrices with i.i.d elements. Moreover, random Schr$\ddot{\hbox{o}}$dinger matrices with parameter
$\lambda$ and RBM with the width of the band $W$ are expected to have some similar qualitative
properties when $\lambda\approx W^{-1}$
(for more details on these conjectures see \cite{Sp:12}).

The key physical parameter of these models is the localization length,
which describes the typical length scale of the eigenvectors of random matrices.
The system is called delocalized if the localization length $\ell$ is
comparable with the matrix size, and it is called localized otherwise.
Delocalized systems correspond to electric
conductors, and localized systems are insulators.

In the case of 1D RBM there is a physical conjecture (see \cite{Ca-Co:90, FM:91}) stating that $\ell$ is
of order $W^2$ (for the energy in the bulk of the spectrum), which means that varying $W$
we can see the crossover: for $W\gg \sqrt{N}$ the eigenvectors are expected to be
delocalized and for $W\ll \sqrt{N}$ they are localized. In terms of eigenvalues
this means that the local eigenvalue statistics in the bulk of the spectrum changes from Poisson,
for $W\ll \sqrt{N}$, to GUE
(Hermitian matrices with i.i.d Gaussian elements),
for $W\gg \sqrt{N}$. At the present time only some upper and lower
bounds for $\ell$ are proven rigorously. It is known from the paper \cite{S:09} that $\ell\le W^8$.
On the other side, in the resent papers \cite{EK:11, Yau:12} it was proven first that $\ell\gg W^{7/6}$,
and then that
$\ell\gg W^{5/4}$.

The questions of the order of the localization length are closely related to the universality conjecture
of the bulk local regime of the random matrix theory, which we briefly outline now.

Let $\lambda_1^{(N)},\ldots,\lambda_N^{(N)}$ be the eigenvalues of
$H_N$. Define their Normalized Counting Measure
(NCM) as
\begin{equation} \label{NCM}
\mathcal{N}_N(\sigma)=\sharp\{\lambda_j^{(N)}\in
\sigma,j=1,\ldots,N \}/N,\quad \mathcal{N}_N(\mathbb{R})=1,
\end{equation}
where $\sigma$ is an arbitrary interval of the real axis.
The behavior of $\mathcal{N}_N$ as $N\to\infty$ was studied for many ensembles.
For 1D RBM it was shown
in \cite{BMP:91, MPK:92} that $\mathcal{N}_{N}$ converges weakly, as $N,W\to\infty$, to a non-random measure
$\mathcal{\mathcal{N}}$, which is called the limiting NCM of the ensemble. The measure $\mathcal{N}$ is absolutely continuous
and its density $\rho$ is given by the well-known Wigner semicircle law (the same result
is valid for Wigner ensembles, in particular, for Gaussian ensembles GUE, GOE):
\begin{equation}\label{rho_sc}
\rho(\lambda)=\dfrac{1}{2\pi}\sqrt{4-\lambda^2},\quad \lambda\in[-2,2].
\end{equation}
Much more delicate result about the density of states at arbitrarily short scales is proven in
\cite{DPS:02} for $d=3$.

These results characterize the so-called global distribution of the eigenvalues.

The local regime deals with the behavior of eigenvalues of $N\times N$
random matrices on the intervals whose length is of the order of the mean distance between
nearest eigenvalues. The main objects of the local regime are $k$-point correlation functions
$R_k$ ($k=1,2,\ldots$), which can be defined by the equalities:
\begin{multline} \label{R}
\mathbf{E}\left\{ \sum_{j_{1}\neq ...\neq j_{k}}\varphi_k
(\lambda_{j_{1}}^{(N)},\dots,\lambda_{j_{k}}^{(N)})\right\}\\ =\int_{\mathbb{R}^{k}} \varphi_{k}
(\lambda_{1}^{(N)},\ldots,\lambda_{k}^{(N)})R_{k}(\lambda_{1}^{(N)},\ldots,\lambda_{k}^{(N)})
d\lambda_{1}^{(N)}\ldots d\lambda_{k}^{(N)},
\end{multline}
where $\varphi_{k}: \mathbb{R}^{k}\rightarrow \mathbb{C}$ is
bounded, continuous and symmetric in its arguments and the
summation is over all $k$-tuples of distinct integers $
j_{1},\dots,j_{k}\in\{1,\ldots,N\}$.

According to the Wigner -- Dyson universality conjecture (see e.g. \cite{Me:91}), the local behavior
of the eigenvalues does not depend on the matrix probability
law (ensemble) and is determined only by the symmetry type of matrices (real
symmetric, Hermitian, or quaternion real in the case of real eigenvalues and orthogonal,
unitary or symplectic in the case of eigenvalues on the unit circle).
For example, the conjecture states that for Hermitian random matrices in the bulk of the spectrum
and in the range of parameters for which the eigenvectors are
delocalized
\begin{multline}\label{Un}
\lim\limits_{N\to \infty}\displaystyle\frac{1}{(N\rho(\lambda_0))^k}
R_k\left(\lambda_0+\displaystyle\frac{\xi_1}{\rho(\lambda_0)\,N},
\ldots,\lambda_0+\displaystyle\frac{\xi_k}{\rho(\lambda_0)\,N}\right)\\
=\det \Big\{\dfrac{\sin \pi(\xi_i-\xi_j)}
{\pi(\xi_i-\xi_j)}\Big\}_{i,j=1}^k
\end{multline}
for any fixed $k$, and the limit is uniform
in $\xi_1, \xi_2,\ldots, \xi_k$ varying
in any compact set in $\mathbb{R}$. This means that the limit coincides with that for GUE.

In this language the conjecture about the crossover for 1D RBM states that we get (\ref{Un})
for $W\gg \sqrt{N}$ (which corresponds to delocalized states), and
we get another behavior, which is determined by the Poisson statistics, for $W\ll \sqrt{N}$
(and corresponds to localized states). For the general Hermitian Wigner matrices (i.e. $W=N$)
bulk universality (\ref{Un}) has been proved recently in \cite{EYY:10, TV:11}. However, in the general case
of RBM the question of bulk universality of local spectral statistics
is still open even for $d=1$.

Other more simple objects of the local regime of the random matrix theory are the
correlation functions (or the mixed moments) of characteristic polynomials.

Characteristic polynomials of random matrices have been actively studied
in the last years (see e.g.
\cite{BaD:03,Br-Hi:01,BorSt:06,St-Fy:03_1,Got-K:08,HKC:01,K-Sn:00,Me-Nor:01,TSh:11,TSh:11_1,St-Fy:03,Va:03}).
 The interest to this topic is stimulated by its connections to the
number theory, quantum chaos, integrable systems, combinatorics, representation
theory and others.

An additional source of motivation for the current work is the development of the supersymmetric method
(SUSY) in the context of random operators with non-trivial spatial structures. This method is widely used in
the physics literature (see e.g. \cite{Ef, M:00}) and is potentially very powerful but the rigorous
control of the
integral representations, which can be obtained by this method, is difficult and so far for the band matrices it
has been performed only for the density of states (see \cite{DPS:02}). From the SUSY point of view
characteristic polynomials correspond to the so-called fermionic sector of the supersymmetric full model,
which describes the correlation functions $R_k$. So the analysis of the local regime of correlation functions
of the characteristic polynomial is an important step
towards the proof of (\ref{Un}).

The correlation function of the characteristic polynomials is
\begin{equation}\label{F}
F_{2k}(\Lambda)=\intd \prod\limits_{s=1}^{2k}\mdet(\lambda_s-H_N)P_n(d\,H_N),
\end{equation}
where $P_n(d\,H_N)$ is defined in (\ref{band}),
and $\Lambda=\hbox{diag}\,\{\lambda_1,\ldots,\lambda_{2k}\}$ are real or complex parameters
that may depend on $N$.

The asymptotic local behavior in the bulk of the spectrum of the $2k$-point mixed moment
for GUE is well-known (see e.g. \cite{St-Fy:03}):
\begin{multline}\label{b_GUE}
F_{2k}\left(\Lambda_0+\hat{\xi}/N\rho(\lambda_0)\right)\\
=C_N\dfrac{\mdet
\Big\{\dfrac{\sin\pi(\xi_i-\xi_{j+k})}{\pi(\xi_i-\xi_{j+k})}
\Big\}_{i,j=1}^k}{\triangle(\xi_1,\ldots,\xi_k)\triangle(\xi_{k+1},\ldots,\xi_{2k})}\times
e^{\lambda_0(\xi_1+\ldots+\xi_{2k})/2\rho(\lambda_0)}(1+o(1)),
\end{multline}
where $\triangle(\xi_1,\ldots,\xi_k)$ is the
Vandermonde determinant of $\xi_1,\dots, \xi_k$, $\hat{\xi}=\hbox{diag}\,\{\xi_1,\ldots,\xi_{2k}\}$,
$\Lambda_0=\lambda_0\cdot I$.

The similar result for the $\beta$-ensembles with $\beta=2$ was obtained in \cite{Br-Hi:00, St-Fy:03}
(see the reference for the more precise statement). In the case of general Hermitian Wigner matrices it was
proven that constant $C_N$ depends only on the first four moments of the matrix elements distribution and does not
depend on any higher moments (see \cite{Got-K:08} for the case $k=1$ and \cite{TSh:11} for any $k$).
The same result was obtained for general Hermitian sample covariance matrices (see \cite{Kos:09} for
the case $k=1$ and \cite{TSh:11_1} for any $k$).
This
shows that the local regime of correlation functions of characteristic polynomials is universal up to the
first four moments.

In this paper we are interested in the asymptotic behavior of (\ref{F}) with $k=1$ for matrices (\ref{ban})
-- (\ref{band}) as $N,W\to\infty$, $W^2=N^{1+\theta}$, $0<\theta\le 1$ (i.e. $W\gg \sqrt{N}$), and for
\begin{equation*}
\lambda_j=\lambda_0+\dfrac{\xi_j}{N\rho(\lambda_0)},\quad j=1,2,
\end{equation*}
where $N=2n+1$, $\lambda_0\in (-2,2)$, $\rho$ is defined in (\ref{rho_sc}), and
$\widehat{\xi}=\hbox{diag}\,\{\xi_1,\xi_2\}$ are real parameters varying in any compact
set $K\subset \mathbb{R}$.

Set also
\begin{equation}\label{D_2}
D_2=\prod\limits_{l=1}^2F_2^{1/2}\Big(\lambda_0+\dfrac{\xi_l}{N\rho(\lambda_0)},\lambda_0+
\dfrac{\xi_l}{N\rho(\lambda_0)}\Big).
\end{equation}

 The main result of the paper is the following theorem:
\begin{theorem}\label{thm:1}
Consider the random matrices (\ref{ban}) -- (\ref{band}) with $W^2=N^{1+\theta}$, where $0<\theta\le 1$.
Define the second mixed moment $F_2$ of the characteristic polynomials as in (\ref{F}). Then we have
\begin{equation}\label{lim1}
\lim\limits_{n\to\infty}
D_2^{-1}F_{2}\Big(\Lambda_0+\hat{\xi}/(N\rho(\lambda_0))\Big)
=\dfrac{\sin(\pi(\xi_1-\xi_{2}))}{\pi(\xi_1-\xi_2)},
\end{equation}
and the limit is uniform in $\xi_1, \xi_2$ varying in any compact set $K\subset\mathbb{R}$. Here
$\rho(\lambda)$ and $D_2$ are defined in (\ref{rho_sc}) and (\ref{D_2}),
$\Lambda_0=\mathrm{diag}\,\{\lambda_0,\lambda_0\}$,
$\lambda_0\in (-2,2)$, $\hat{\xi}=\mathrm{diag}\,\{\xi_1,\xi_2\}$.
\end{theorem}
The theorem shows that the limit above for the second mixed moment of characteristic polynomials for
1D Gaussian random band matrices (with $W^2=N^{1+\theta}$, $0<\theta\le 1$) coincides with that for the
Gaussian unitary ensemble, i.e. the local behavior of the second mixed moment in the bulk of the spectrum is
universal. In the case $W\ll \sqrt{N}$ the limit is expected to be different from (\ref{lim1}), but we will not
discuss it in this paper.

The paper is organized as follows. In Section $2$ we obtain a convenient integral
representation for $F_{2}$,
using the integration
over the Grassmann variables. The method is a generalization of that of \cite{Br-Hi:00,Br-Hi:01} and is an
analog of the method of \cite{TSh:11,TSh:11_1}, where the Hermitian Wigner and general sample covariance matrices
were considered.
In Section $3$ we give the sketch of the proof of Theorem
\ref{thm:1}. Section $4$ deals with the most important preliminary results needed for the proof.
In Section $5$ we prove Theorem \ref{thm:1}, applying the steepest descent method
to the integral representation. Section 6 is devoted to the proofs of the auxiliary statements.

\subsection{Notation}
We denote by $C$, $C_1$, etc. various $W$ and $N$-independent quantities below, which
can be different in different formulas. Integrals
without limits denote the integration (or the multiple integration) over the whole
real axis, or over the Grassmann variables.

Moreover,
\begin{itemize}
\item $N=2n+1;$

\item $J=(-W^2\Delta+1)^{-1};$

\item $\mathbf{E}\big\{\ldots\big\}$ is an expectation with respect to the measure (\ref{band});

\item $U_\varepsilon(x)=(x-\varepsilon,x+\varepsilon)\subset \mathbb{R};$
\end{itemize}

\begin{fleqn}[16pt]
 \begin{equation}\label{a_pm}
\bullet \,\,\, a_{\pm}=\pm\dfrac{\sqrt{4-\lambda_0^2}}{2}=\pm\pi\rho(\lambda_0),\quad \overline{a}_\pm
=(a_\pm,\ldots,a_\pm)\in \mathbb{R}^N,
\end{equation}
\quad \quad \,where $\rho$ is defined in (\ref{rho_sc});

\begin{itemize}
\item $
\Lambda_0=\left(\begin{array}{cc}
\lambda_0&0\\
0& \lambda_0
\end{array}\right),\quad
\Lambda=\left(\begin{array}{cc}
\lambda_1&0\\
0& \lambda_2
\end{array}\right),\quad
\hat{\xi}=\left(\begin{array}{cc}
\xi_1&0\\
0& \xi_2
\end{array}\right),\quad
L=\left(\begin{array}{cc}
a_+&0\\
0& a_-
\end{array}\right);
$

\item $d\mu$ is the Haar measure on $U(2)$;

\end{itemize}

\begin{align}\label{f}
\bullet \,\,\,\, &f(x)=(x+i\lambda_0/2)^2/2-\log(x-i\lambda_0/2),\\ \notag
\quad\, &f_*(x)=\Re(f(x)-f(a_\pm));
\end{align}

\begin{itemize}
\item $\Omega_\delta$ is a union of
\begin{align} \label{Om_delta}
 \Omega^+_\delta&=\{ \{a_j\}, \{b_j\}: a_j, b_j\in U_\delta(a_+) \,\,\forall j\},\\ \notag
 \Omega^-_\delta&=\{ \{a_j\}, \{b_j\}: a_j, b_j\in U_\delta(a_-)\,\,\forall j\},\\ \notag
\Omega^\pm_\delta&=\{ \{a_j\}, \{b_j\}: (a_j\in U_\delta(a_+),\,\, b_j\in U_\delta(a_-))\\ \notag
&\quad\quad\quad\quad\quad\quad\quad\quad\quad\quad\quad\quad\quad\,\,\hbox{or}\,\,
(a_j\in U_\delta(a_-),\, b_j\in U_\delta(a_+))\,\,\forall j\},
\end{align}
where $\delta=W^{-\kappa}$ and $\kappa<\theta/8$.

\item $\Sigma$ is an integral over $\Omega_\delta$, $\Sigma_c$ is an integral over its complement,
and $\Sigma_\pm$, $\Sigma_+$ and $\Sigma_-$ are integrals over $\Omega^\pm_\delta$,
$\Omega^+_\delta$ and $\Omega^-_\delta$.
\end{itemize}

\begin{align} \label{c_pm}
\bullet \,\,\, c_\pm&=1-\dfrac{\lambda_0^2}{4}\pm
\dfrac{i\lambda_0}{2}\cdot \sqrt{1-\lambda_0^2/4},\quad c_0=\Re f(a_+);
\end{align}

\begin{align}\label{mu}
\bullet \,\,\, \mu_{\gamma}(x)=\exp\big\{-\frac{1}{2}\sum\limits_{j=-n+1}^{n}(x_j-x_{j-1})^2
-\frac{\gamma}{W^2}\sum\limits_{j=-n}^{n}x_j^2\big\};
\end{align}

\begin{align}\label{angle}
\bullet \,\,\, &\langle \ldots \rangle_0=Z_{\delta,\gamma}^{-1}\intd_{-\delta W}^{\delta W} (\ldots) \cdot
\mu_{\gamma}(x)\prod\limits_{q=-n}^nd x_q,& Z_{\delta,\gamma}&=
\intd_{-\delta W}^{\delta W} \mu_{\gamma}(x) \prod\limits_{q=-n}^nd x_q,\\ \notag
&\langle \ldots \rangle=Z^{-1}_\gamma\intd (\ldots) \cdot \mu_{\gamma}(x) \prod\limits_{q=-n}^nd x_q, & Z_\gamma&=
\intd \mu_{\gamma}(x) \prod\limits_{q=-n}^nd x_q,
\end{align}
\end{fleqn}
\quad \quad \,where $\delta>0$ and $\gamma\in \mathbb{C}$, $\Re\gamma>0$;
\begin{itemize}
\item $\langle \ldots \rangle_*$
(and $\langle \ldots \rangle_{0,*}$) is (\ref{angle}) with $\mu_{\Re \gamma}(x)$
instead of $\mu_{\gamma}(x)$.
\end{itemize}

\section{Integral representation}
In this section we obtain an integral representation for $F_{2}$ of (\ref{F}) by using integration
over the Grassmann variables. This method allows us to obtain
the formula for the product of characteristic
polynomials, which is very useful for the
averaging because it is a Gaussian-type integral (see the formula (\ref{G_Gr}) below).
After averaging over the probability measure we can integrate over the Grassmann variables to get
an integral representation (in complex variables) which can be studied by the steepest descent method.

Integration over the Grassmann variables has been introduced by Berezin and is widely used in the physics
literature (see e.g.
\cite{Ber, Ef, M:00}). For the reader's convenience we give a brief outline of the techniques.

\subsection{Grassmann integration}
Let us consider two sets of formal variables
$\{\psi_j\}_{j=1}^n,\{\overline{\psi}_j\}_{j=1}^n$, which satisfy the anticommutation
conditions
\begin{equation}\label{anticom}
\psi_j\psi_k+\psi_k\psi_j=\overline{\psi}_j\psi_k+\psi_k\overline{\psi}_j=\overline{\psi}_j\overline{\psi}_k+
\overline{\psi}_k\overline{\psi}_j=0,\quad j,k=1,\ldots,n.
\end{equation}
Note that this definition implies $\psi_j^2=\overline{\psi}_j^2=0$.
These two sets of variables $\{\psi_j\}_{j=1}^n$ and $\{\overline{\psi}_j\}_{j=1}^n$ generate the Grassmann
algebra $\mathfrak{A}$. Taking into account that $\psi_j^2=0$, we have that all elements of $\mathfrak{A}$
are polynomials of $\{\psi_j\}_{j=1}^n$ and $\{\overline{\psi}_j\}_{j=1}^n$ of degree at most one
in each variable. We can also define functions of
the Grassmann variables. Let $\chi$ be an element of $\mathfrak{A}$, i.e.
\begin{equation}\label{chi}
\chi=a+\sum\limits_{j=1}^n (a_j\psi_j+ b_j\overline{\psi}_j)+\sum\limits_{j\ne k}
(a_{j,k}\psi_j\psi_k+
b_{j,k}\psi_j\overline{\psi}_k+
c_{j,k}\overline{\psi}_j\overline{\psi}_k)+\ldots.
\end{equation}
For any
sufficiently smooth function $f$ we define by $f(\chi)$ the element of $\mathfrak{A}$ obtained by substituting $\chi-a$
in the Taylor series of $f$ at the point $a$. Since $\chi$ is a polynomial of $\{\psi_j\}_{j=1}^n$,
$\{\overline{\psi}_j\}_{j=1}^n$ of the form (\ref{chi}), according to (\ref{anticom}), there exists such
$l$ that $(\chi-a)^l=0$, and hence the series terminates after a finite number of terms, and so $f(\chi)\in \mathfrak{A}$.

For example, we have
\begin{align}\notag
 &\exp\{a\,\overline{\psi}_1\psi_1\}=1+a\,\overline{\psi}_1\psi_1+(a\,\overline{\psi}_1\psi_1)^2/2+\ldots
 =1+a\,\overline{\psi}_1\psi_1,\\ \notag
&\exp\{a_{11}\overline{\psi}_1\psi_1+a_{12}\overline{\psi}_1\psi_2+
a_{21}\overline{\psi}_2\psi_1+a_{22}\overline{\psi}_2\psi_2\}=1+ a_{11}\overline{\psi}_1\psi_1\\
\label{ex_12} &+a_{12}\overline{\psi}_1\psi_2+ a_{21}\overline{\psi}_2\psi_1+a_{22}\overline{\psi}_2\psi_2+
(a_{11}\overline{\psi}_1\psi_1+a_{12}\overline{\psi}_1\psi_2\\ \notag &+
a_{21}\overline{\psi}_2\psi_1+a_{22}\overline{\psi}_2\psi_2)^2/2+\ldots=1+
a_{11}\overline{\psi}_1\psi_1+a_{12}\overline{\psi}_1\psi_2+ a_{21}\overline{\psi}_2\psi_1\\ \notag
&+a_{22}\overline{\psi}_2\psi_2+(a_{11}a_{22}-a_{12}a_{21})\overline{\psi}_1\psi_1\overline{\psi}_2\psi_2.
\end{align}
Following Berezin \cite{Ber}, we define the operation of
integration with respect to the anticommuting variables in a formal
way:
\begin{equation}\label{int_gr}
\intd d\,\psi_j=\intd d\,\overline{\psi}_j=0,\quad \intd
\psi_jd\,\psi_j=\intd \overline{\psi}_jd\,\overline{\psi}_j=1,
\end{equation}
and then extend the definition to the general element of $\mathfrak{A}$ by
the linearity. A multiple integral is defined to be a repeated
integral. Assume also that the ``differentials'' $d\,\psi_j$ and
$d\,\overline{\psi}_k$ anticommute with each other and with the
variables $\psi_j$ and $\overline{\psi}_k$.

Thus, according to the definition, if
$$
f(\psi_1,\ldots,\psi_k)=p_0+\sum\limits_{j_1=1}^k
p_{j_1}\psi_{j_1}+\sum\limits_{j_1<j_2}p_{j_1j_2}\psi_{j_1}\psi_{j_2}+
\ldots+p_{1,2,\ldots,k}\psi_1\ldots\psi_k,
$$
then
\begin{equation}\label{int}
\intd f(\psi_1,\ldots,\psi_k)d\,\psi_k\ldots d\,\psi_1=p_{1,2,\ldots,k}.
\end{equation}

 Let $A$ be an ordinary matrix with a positive Hermitian part. The following Gaussian
integral is well-known:
\begin{equation}\label{G_C}
\intd \exp\Big\{-\sum\limits_{j,k=1}^nA_{j,k}z_j\overline{z}_k\Big\} \prod\limits_{j=1}^n\dfrac{d\,\Re
z_jd\,\Im z_j}{\pi}=\dfrac{1}{\mdet A}.
\end{equation}
One of the important formulas of the Grassmann variables theory is the analog of (\ref{G_C}) for the
Grassmann variables (see \cite{Ber}):
\begin{equation}\label{G_Gr}
\int \exp\Big\{-\sum\limits_{j,k=1}^nA_{j,k}\overline{\psi}_j\psi_k\Big\}
\prod\limits_{j=1}^nd\,\overline{\psi}_jd\,\psi_j=\mdet A,
\end{equation}
where $A$ now is any $n\times n$ matrix.

For $n=1$ and $n=2$ this formula follows immediately from (\ref{ex_12}) and (\ref{int}).

Also we will need the Hubbard-Stratonovich transform (see e.g. \cite{Sp:12}).
This is a well-known simple trick, which
is just the Gaussian integration. In the simplest form it looks as following:
\begin{equation}\label{Hub}
e^{a^2/2}=(2\pi)^{-1/2} \int e^{-x^2/2+ax}dx.
\end{equation}
Here $a$ can be complex number or the sum of the products of even numbers of Grassmann variables.


\subsection{Formula for $F_{2}$}
\begin{lemma}\label{l:int_repr}
The second mixed moment of the characteristic polynomials for 1D Hermitian Gaussian band
matrices, defined in (\ref{F}), can be represented as follows:
\begin{align}\label{F_rep}
&F_2\Big(\Lambda_0+\dfrac{\hat{\xi}}{N\rho(\lambda_0)}\Big)=-(2\pi^2)^{-N}\mdet^{-2} J
\int\exp\Big\{-\frac{W^2}{2}\sum\limits_{j=-n+1}^n\Tr
(X_j-X_{j-1})^2\Big\}\\ \notag
&\times\exp\Big\{-
\frac{1}{2}\sum\limits_{j=-n}^n \Tr\Big(X_j+\frac{i\Lambda_0}{2}+
\frac{i\hat{\xi}}{N\rho(\lambda_0)}\Big)^2\Big\}\prod\limits_{j=-n}^n
\det\big(X_j-i\Lambda_0/2\big)\prod\limits_{j=-n}^ndX_j,
\end{align}
where $\{X_j\}$ are $2\times 2$ Hermitian matrices and
\begin{equation}\label{dX}
dX_j=d(X_{j})_{11}d(X_{j})_{22}d\Re (X_{j})_{12}d\Im (X_{j})_{12}.
\end{equation}
Moreover, this formula can be rewritten in the form
\begin{align}\label{F_0_1}
F_2\Big(\Lambda_0+&\dfrac{\hat{\xi}}{N\rho(\lambda_0)}\Big)=-\dfrac{C(\xi)\mdet^{-2}J}{(4\pi)^{N}}
\intd\limits\exp\Big\{-\frac{W^2}{2}\sum\limits_{j=-n+1}^n\Tr (V_j^*A_jV_j-A_{j-1})^2\Big\}\\ \notag
&\times \exp\Big\{-\sum\limits_{j=-n}^n(f(a_j)+f(b_j))-
\frac{i}{N\rho(\lambda_0)}\sum\limits_{j=-n}^n\Tr \big(P_jU_{-n}\big)^*A_j\,
(P_jU_{-n}\big)\hat{\xi}\Big\}
\\ \notag
&\times\prod\limits_{l=-n}^n(a_l-b_l)^2d\,\mu(U_{-n})\,d\overline{a\vphantom{b}}\, d\overline{b}\,
\prod\limits_{q=-n+1}^nd\mu(V_q),
\end{align}
where $f$ is defined in (\ref{f}),
$A_j=\mathrm{diag}\{a_j,b_j\}$, $\{V_j\}$ and $U_{-n}$ are $2\times 2$ unitary matrices,
$d\mu(U)$ is the Haar measure on $U(2)$, and
\begin{equation}\label{P_k}
P_k=\prod\limits_{s=k}^{-n+1}V_s,\quad C(\xi)=\exp\Big\{\dfrac{\lambda_0(\xi_1+\xi_2)}{2\rho(\lambda_0)}+
\dfrac{\xi_1^2+\xi_2^2}{2N\rho(\lambda_0)^2}\Big\}.
\end{equation}
\end{lemma}
\begin{remark}
Formula (\ref{F_rep}) is valid for any dimension if we change the sum $\sum\Tr
(X_j-X_{j-1})^2$ to $\sum\Tr
(X_j-X_{j^\prime})^2$, where the last sum runs over all pairs of nearest neighbor
$j, j^\prime$ in the volume $\mathcal{L}\subset \mathbb{Z}^d$ (see the definition of RBM
(\ref{ban}) -- (\ref{J})).
\end{remark}
\begin{proof}
Using (\ref{G_Gr}) we obtain
\begin{equation}\label{ne_usr}
\begin{array}{c}
F_2(\Lambda)={\bf E}\bigg\{\displaystyle\int
e^{-\sum\limits_{\alpha=1}^2\sum\limits_{j,k=-n}^n(\lambda_\alpha-H_N)_{jk}
\overline{\psi}_{j\alpha}\psi_{k\alpha}}\prod\limits_{\alpha=1}^2\prod\limits_{q=-n}^n
d\,\overline{\psi}_{q\alpha}d\,\psi_{q\alpha}\bigg\}\\
={\bf E}\bigg\{\displaystyle\int e^{-\sum\limits_{\alpha=1}^2\lambda_\alpha\sum\limits_{p=-n}^n
\overline{\psi}_{p\alpha}\psi_{p\alpha}} \exp\bigg\{\sum\limits_{j<k}\sum\limits_{\alpha=1}^2\Big(\Re
H_{jk}\cdot(\overline{\psi}_{j\alpha}\psi_{k\alpha}
+\overline{\psi}_{k\alpha}\psi_{j\alpha})\\
+i\Im H_{jk}\cdot(\overline{\psi}_{j\alpha}\psi_{k\alpha}-\overline{\psi}_{k\alpha}
\psi_{j\alpha})\Big)+\sum\limits_{j=-n}^nH_{jj}\cdot\sum\limits_{\alpha=1}^2\overline{\psi}_{j\alpha}
\psi_{j\alpha}\bigg\}
\prod\limits_{\alpha=1}^2\prod\limits_{q=-n}^nd\,\overline{\psi}_{q\alpha}d\,\psi_{q\alpha}\bigg\},
\end{array}
\end{equation}
where $\{\psi_{j\alpha}\}$, $j=-n,\ldots,n$, $\alpha=1,2$ are the Grassmann variables ($2n+1$ variables for each
determinant in (\ref{F})). Here and below we use Greek letters like $\alpha, \beta$ etc. for
the field index and Latin letters $j, k$ etc. for the position index.

Integrating over the measure (\ref{band}) we get
\begin{align}\label{usr}
&F_2(\Lambda)=\displaystyle\int \prod\limits_{\alpha=1}^2\prod\limits_{q=-n}^nd\,\overline{\psi}_{q\alpha}
d\,\psi_{q\alpha}
\exp\Big\{-\sum\limits_{\alpha=1}^2\lambda_\alpha\sum\limits_{p=-n}^n \overline{\psi}_{p\alpha}
\psi_{p\alpha}\Big\}\\ \notag\times
\exp&\Big\{\sum\limits_{j<k}J_{jk}(\overline{\psi}_{j1}\psi_{k1}+\overline{\psi}_{j2}\psi_{k2})
(\overline{\psi}_{k1}\psi_{j1}+\overline{\psi}_{k2}\psi_{j2})+\sum\limits_{j=-n}^n\dfrac{J_{jj}}{2}
(\overline{\psi}_{j1}\psi_{j1}+\overline{\psi}_{j2}\psi_{j2})^2\Big\}.
\end{align}
Applying a couple of times the Hubbard-Stratonovich transform (\ref{Hub}), we get:
\begin{multline}\label{HS}
\intd\exp\Big\{-\dfrac{1}{2}\sum\limits_{jk}J^{-1}_{jk}\Tr
X_jX_k-i\sum\limits_{j}(\overline{\psi}_{j1},\overline{\psi}_{j2})X_j
\binom{\psi_{j1}}{\psi_{j2}}\Big\}\prod\limits_{j=-n}^ndX_j\\
=(2\pi^2)^{N}\mdet^2 J\cdot
\exp\Big\{\dfrac{1}{2}\sum\limits_{j,k}J_{jk}(\overline{\psi}_{j1}\psi_{k1}+\overline{\psi}_{j2}\psi_{k2})
(\overline{\psi}_{k1}\psi_{j1}+\overline{\psi}_{k2}\psi_{j2})\Big\},
\end{multline}
where $X_j$ is Hermitian $2\times 2$ matrix and $dX_j$ is defined in (\ref{dX}).

Substituting this and (\ref{J}) for $J^{-1}_{jk}$ into (\ref{usr}), putting
$\Lambda=\Lambda_0+\dfrac{\hat{\xi}}{N\rho(\lambda_0)}$, and using (\ref{G_Gr}) to
integrate over the Grassmann variables, we obtain
\begin{align*}\notag
&F_2\Big(\Lambda_0+\dfrac{\hat{\xi}}{N\rho(\lambda_0)}\Big)=(2\pi^2)^{-N}\mdet^{-2} J
\int\exp\Big\{-\frac{W^2}{2}\sum\limits_{j=-n+1}^n\hbox{Tr}\,(X_j-X_{j-1})^2\Big\}\\
&\times \exp\Big\{-
\frac{1}{2}\sum\limits_{j=-n}^n\Tr X_j^2\Big\}\prod\limits_{j=-n}^n
\det\big(iX_j+\Lambda_0+\hat{\xi}/N\rho(\lambda_0)\big)\prod\limits_{j=-n}^ndX_j\\
&=-(2\pi^2)^{-N}\mdet^{-2} J\int\exp\Big\{-\frac{W^2}{2}\sum\limits_{j=-n+1}^n\hbox{Tr}\,(X_j-X_{j-1})^2-
\frac{1}{2}\sum\limits_{j=-n}^n\Tr X_j^2\Big\}\\
&\times\prod\limits_{j=-n}^n
\det\big(X_j-i\Lambda_0-i\hat{\xi}/N\rho(\lambda_0)\big)\prod\limits_{j=-n}^ndX_j,
\end{align*}
which gives (\ref{F_rep}) after shifting $X_j\to X_j+i\Lambda_0/2+i\hat{\xi}/N\rho(\lambda_0)$. The reason
of such a shift is that we need to have saddle-points lying on the contour of the integration (see (\ref{a_pm})
below).

Let us change the variables to $X_j=U_j^*A_jU_j$, where $U_j$ is a unitary matrix and
$A_j=\hbox{diag}\,\{a_j,b_j\}$, $j=-n,\ldots,n$. Then $d X_j$ of (\ref{dX}) becomes (see e.g. \cite{Me:91}, Section 3.3)
$$\dfrac{\pi}{2}(a_j-b_j)^2da_j\,db_j
d\mu(U_j),$$ where $d\mu(U_j)$ is the normalized to unity Haar measure on the unitary group $U(2)$. Thus, we have
\begin{align*}
F_2\Big(\Lambda_0+\dfrac{\hat{\xi}}{N\rho(\lambda_0)}\Big)=&-\dfrac{C(\xi)\mdet^{-2}J }{(4\pi)^{N}}
\int\,d\overline{a\vphantom{b}}\, d\overline{b}\int\limits_{U(2)^N}\,\prod\limits_{j=-n}^nd\mu(U_j)\\
&\times\exp\Big\{-\frac{W^2}{2}\sum\limits_{j=-n+1}^n\Tr (U_j^*A_jU_j-
U_{j-1}^*A_{j-1}U_{j-1})^2\Big\}\\
&\times
\exp\Big\{-\frac{1}{2}\sum\limits_{j=-n}^n\Tr \Big(A_j+\frac{i\Lambda_0}{2}\Big)^2-
\dfrac{i}{N\rho(\lambda_0)}\sum\limits_{j=-n}^n\Tr U_j^*A_jU_j\hat{\xi}\Big\}
\\
&\times\prod\limits_{k=-n}^n (a_k-i\lambda_0/2\big) (b_k-i\lambda_0/2\big)\prod\limits_{k=-n}^n
(a_k-b_k)^2,
\end{align*}
where
\begin{equation}\label{da}
d\overline{a\vphantom{b}}=\prod\limits_{j=-n}^nda_j,\quad d\overline{b}=\prod\limits_{j=-n}^ndb_j,\quad
C(\xi)=\exp\{\lambda_0(\xi_1+\xi_2)/2\rho(\lambda_0)\}.
\end{equation}
Now changing the ``angle variables'' $U_j$ to $V_j=U_jU_{j-1}^*$, $j=-n+1,\ldots,n$ (i.e. the new variables are
$U_{-n}, V_{-n+1}, V_{-n+2},\ldots,V_n$), we get (\ref{F_0_1}).
\end{proof}

\section{Sketch of the proof of Theorem \ref{thm:1}}

The strategy of the proof is the following.

\smallskip

First we will study the function $f$ and find that expected saddle-points for each $a_j$ and $b_j$
are $a_\pm$, which are defined in (\ref{a_pm}). This will be done in Section 4.1.

The second step is to prove that the main contribution to the integral (\ref{F_0_1}) is given
by $\Sigma$, i.e. by the integral over $\Omega_\delta$ (see (\ref{Om_delta})). More precisely, we are going to prove that
\begin{equation}\label{F_2}
F_2\Big(\Lambda_0+\dfrac{\hat{\xi}}{N\rho(\lambda_0)}\Big)=-\dfrac{C(\xi) \mdet^{-2}J}{(4\pi)^{N}}
\cdot \Sigma\cdot(1+o(1)),\quad W\to \infty.
\end{equation}
The bound for the complement $|\Sigma_c|$ can be obtained by inserting the absolute value inside the integral and
by performing exactly the integral over the unitary groups. After this, since we are
far from the saddle-points of $f$, one can control the integral. This will be done in Lemma \ref{l:2},
Section 5.1.

The next step is the calculation of $\Sigma$ (see Section 5.2, Lemma \ref{l:sigma}).
First note that shifting
$$U_j \to \left(
\begin{array}{ll}
0&1\\
1&0
\end{array}
\right) U_j$$
for some $j$, we can rotate each domain of type $$\{ \{a_j\}, \{b_j\}: (a_j\in U_\delta(a_+),\,\,
b_j\in U_\delta(a_-))\,\,\hbox{or}\,\,
(a_j\in U_\delta(a_-),\, b_j\in U_\delta(a_+))\,\,\forall j\}$$ to the $\delta$-neighborhood
of the point $(\overline{a}_+,\overline{a}_-)$ with $\overline{a}_\pm$ of (\ref{a_pm}). Thus, we can
consider the contribution over $\Omega^{\pm}_{\delta}$ as
$2^{N}$ contributions of the $\delta$-neighborhood of the point $(\overline{a}_+,\overline{a}_-)$.
Consider this neighborhood (or the neighborhoods of the points
$a_j=b_j=a_+$ or $a_j=b_j=a_-$ for $\Omega^{+}_{\delta}$ or $\Omega^{-}_{\delta}$ correspondingly),
and change the variables as
\begin{align}\label{change}
a_j&\to a_++\tilde{a}_j/W,\quad |\tilde{a}_j|\le\delta W,\\ \notag
b_j&\to a_-+\tilde{b}_j/W, \quad\, |\tilde{b}_j|\le\delta W,
\end{align}
and set $\widetilde{A}_j=\hbox{diag}\,\{\tilde{a\vphantom{b}}_j,\tilde{b}_j\}$.
To compute $\Sigma$, one has to perform first the integral over the unitary groups.
This integral is some analytic in $\{\tilde{a\vphantom{b}}_j/W\}$, $\{\tilde{b}_j/W\}$ function. The main idea is to prove that the leading part of this function can be obtained by
replacing all $V_s$ in the ``bad'' term
$$\exp\Big\{-\frac{i}{N\rho(\lambda_0)}\sum\limits_{j=-n}^n
\Tr \big(\prod\limits^{-n+1}_{s=j} V_s\cdot U_{-n}\big)^*(L+\widetilde{A}_s/W)\,
(\prod\limits^{-n+1}_{s=j} V_s \cdot U_{-n}\big)\hat{\xi}\Big\}$$
with $I$. To this end, we expand the ``bad'' term into the series and for each summand, which is
analytic in $\{\tilde{a\vphantom{b}}_j/W\}$, $\{\tilde{b}_j/W\}$, find the bound for
its Taylor coefficients (see Lemma \ref{l:un}).

To integrate with respect to $\{\tilde{a\vphantom{b}}_j\}$, $\{\tilde{b}_j\}$, we expand
\begin{equation}\label{f_exp}
\begin{array}{cc}
f(x)-f(a_\pm)=c_{\pm}(x-a_\pm)^2+s_3(x-a_\pm)^3+\ldots=c_{\pm}(x-a_\pm)^2+\varphi_\pm(x-a_\pm), \\
 c_\pm=1-\dfrac{\lambda_0^2}{4}\pm
\dfrac{i\lambda_0}{2}\cdot \sqrt{1-\dfrac{\lambda_0^2}{4}},
\end{array}
\end{equation}
and then leave only the quadratic form in the exponent in (\ref{F_0_1}).
At this step we will face with a problem to study a complex valued Gaussian
$(2n+1)$~-dimensional distribution (\ref{mu})
with $\gamma\in \mathbb{C}$, $\Re \gamma>0$ (in our case $\gamma=c_+$ or $c_-$).
The properties of the measure $\mu$ will be studied in Section 4.2.
The most important steps here are:
\begin{itemize}
\item to prove that
\begin{equation}\label{in_sk}
\Big\langle \exp\Big\{\sum\limits_{j=-n}^n\varphi(x_j/W)\Big\} -1\Big\rangle_0=o(1),
\end{equation}
where $\varphi$ is any analytic in the neighborhood of $0$ function whose Taylor expansion
starts from the third order, and
$\langle\ldots\rangle_0$ is defined in (\ref{angle})
(this will be done in Lemma \ref{l:in});

\item to prove that if for some function the absolute value of each coefficient of its Taylor expansion
does not exceed the corresponding coefficient of the other function (majorant), then we can estimate the averaging
of the first function over the complex measure by the averaging of the majorant over the positive
one (see Lemma \ref{l:maj}). This helps to integrate the function obtained after the integration
over the unitary groups, since the proper majorant can be found.
\end{itemize}
These arguments show that the leading term of $\Sigma_\pm$ is the integral over the Gaussian measures
$\mu_{c_\pm}$ in $\{a_j\}$ and $\{b_j\}$
variables, and the integral over the
unitary group $d\mu(U_{-n})$ which gives the sine-kernel. This gives an asymptotic expression for $\Sigma_\pm$
(see Lemma \ref{l:sig},
(\ref{F_0_3}) -- (\ref{F_0_4})).

It will be also shown in Section 5.2.2 that the integrals $\Sigma_+$ and $\Sigma_-$ over $\Omega_\delta^+$ and
$\Omega_\delta^-$ have smaller orders than $\Sigma_\pm$.

\section{Preliminary results}

\subsection{Saddle-point analysis for $f$ of (\ref{f})}
Considering zeros of the first derivative of the function $f$ of (\ref{f}), we find that the expected
saddle-points are $a_\pm$, which are defined in (\ref{a_pm}).

We can write in the small neighborhood of $a_\pm$
\begin{equation*}
\begin{array}{cc}
f(x)-f(a_\pm)=c_{\pm}(x-a_\pm)^2+s_3(x-a_\pm)^3+\ldots=c_{\pm}(x-a_\pm)^2+\varphi_\pm(x-a_\pm), \\
 c_\pm=1-\dfrac{\lambda_0^2}{4}\pm
\dfrac{i\lambda_0}{2}\cdot \sqrt{1-\dfrac{\lambda_0^2}{4}},
\end{array}
\end{equation*}
where $|\varphi_\pm(x-a_\pm)|=O(|x-a_\pm|^3)$.

Let us also study
\begin{align}\label{f*}
f_*(x)&=\Re(f(x)-f(a_\pm))=\frac{1}{2}(x^2-\lambda_0^2/4-\log (x^2+\lambda_0^2/4))-c_0,\\ \notag
c_0&=\Re f(a_\pm)=1/2-\lambda_0^2/4.
\end{align}
We need
\begin{lemma}\label{l:min_L}
The function $f_*(x)$ for $x\in \mathbb{R}$ attains its minimum at
$x=a_{\pm}$, where $a_{\pm}$ is defined in (\ref{a_pm}).
Moreover, $f_*(a_\pm)=0$ and if $x\not\in U_\delta(a_\pm):=(a_\pm-\delta,
a_\pm+\delta)$ for sufficiently small $\delta>0$, then
\begin{equation}\label{ineqv_ReV}
f_*(x)\ge C\delta^2.
\end{equation}
In addition, we have for $x\in (-\infty,\delta)$
\begin{equation}\label{in_left}
f_*(x)\ge \alpha\,(x-a_-)^2,
\end{equation}
where $\alpha$ is some positive constant. A similar inequality holds for $x\in (-\delta,+\infty)$ (with $a_+$ instead
of $a_-$).
\end{lemma}
The proof of this simple lemma can be found in Section 6.

\subsection{Analysis of the measure $\mu_\gamma$}
In this section we study the properties of the complex Gaussian distribution $\mu_\gamma$
 defined in (\ref{mu}). Set
\begin{align}\label{mu_m}
\mu^{(m)}_{\gamma}(x)=\exp\big\{-\frac{1}{2}\sum\limits_{j=2}^{m}(x_j-x_{j-1})^2
-\frac{\gamma}{W^2}\sum\limits_{j=1}^{m}x_j^2\big\}.
\end{align}

\begin{lemma}\label{l:okr}
We have for any $\gamma\in \mathbb{C}$, $\Re \gamma>0$

\begin{description}
\item{(1)\,\,}
 \begin{align}\label{f_eq}
 Z^{(m)}_\gamma&:=\displaystyle\int
\mu^{(m)}_{\gamma}(x)
\prod\limits_{q=1}^md x_q=(2\pi)^{m/2}\mdet^{-1/2}(-\Delta+2\gamma/W^2)\\ \notag
&=(2\pi)^{m/2}\Big(\dfrac{\sqrt{2\gamma}}{W}\sinh \dfrac{m\sqrt{2\gamma}}{W}\Big)^{-1/2}(1+o(1))
\end{align}
Moreover, if we set
\begin{equation}\label{G}
G^{(m)}(\gamma)=\left(-\Delta+\dfrac{2\gamma}{W^2}\right)^{-1},
\end{equation}
then
\begin{equation}\label{G_as}
|G^{(m)}_{ii}(\gamma)|\le\dfrac{C_\gamma W}{\sqrt{2\gamma}}\coth\dfrac{m\sqrt{2\gamma}}{W}(1+o(1)).
\end{equation}


\item{(2)\,\,} $ \dfrac{|Z^{(m)}_\gamma-Z_{\delta,\gamma}^{(m)}|}{|Z^{(m)}_\gamma|}:=
|Z^{(m)}_\gamma|^{-1}\bigg|\,\displaystyle\int\limits_{\max |x_i|>\delta W}\mu^{(m)}_{\gamma}(x)
\prod\limits_{q=1}^md x_q\bigg|\le C_1\, e^{-C_2\delta^2 W},\quad W\to\infty, $ where $m> CW$,
$\delta=W^{-\kappa}$ for sufficiently small $\kappa<\theta/8$, and
\[
Z_{\delta,\gamma}^{(m)}= \intd_{-\delta W}^{\delta W} \mu_{\gamma}^{(m)}(x) \prod\limits_{q=1}^md x_q.
\]

In addition, for any $m$
\[
|Z^{(m)}_\gamma|^{-1}\bigg|\,\displaystyle\int\limits_{|x_k-x_1|>\delta W}\mu^{(m)}_{\gamma}(x)
\prod\limits_{q=1}^md x_q\bigg|\le C_1\, e^{-C_2\delta^2 W},\quad W\to\infty,
\]
and for $m>CW$ and any $\gamma_1,\gamma_2\in \mathbb{C}$, $\Re \gamma_1,\Re \gamma_2>0$
\begin{equation}\label{otn_mod}
\dfrac{|Z^{(m)}_{\gamma_1}|}{|Z^{(m)}_{\gamma_2}|}\le e^{C_1m/W}, \quad W\to\infty.
\end{equation}
\item{(3)\,\,} Let $m>C_1W$, $k\le Cm/W$, $S=\{i_1,\ldots,i_s\} \subset \{1,\ldots,m\}$, and
$\sum\limits_{l=1}^sk_{i_l}=3k$, where
$k_l\in \{3,\ldots,k\}$. Then $$
|Z^{(m)}_\gamma|^{-1}\bigg|\,\displaystyle\int\limits_{\max |x_i|>\delta W}\prod\limits_{j\in S}(x_j/W)^{k_j}\cdot \mu^{(m)}_{\gamma}(x)
\prod\limits_{q=1}^md x_q\bigg|\le e^{-C_1\delta^2 W},\quad W\to\infty,
$$
where $\delta=W^{-\kappa}$ for sufficiently small $\kappa<\theta/8$.
\end{description}
\end{lemma}
The proof of the lemma is rather standard and can be found in Section 6.

Let us study the properties of the averages of (\ref{angle}),
where $\delta=W^{-\kappa}$, $\kappa<\theta/8$.

 We will use bellow the following form of the Wick theorem:
\begin{lemma}\label{l:Wick+}
(i) For any smooth function $f$
\begin{equation}\label{Wick}
\langle x_{i_1} f(x_{i_1},\ldots,x_{i_p}) \rangle=\sum\limits_{j=1}^p\langle x_{i_1}x_{i_j} \rangle
\langle \partial f(x_{i_1},\ldots,x_{i_p})/\partial x_{i_j} \rangle.
\end{equation}
The same is valid for $\langle \ldots \rangle_*$, where $\langle\ldots\rangle$, $\langle\ldots\rangle_*$
are defined in (\ref{angle}).
\bigskip

(ii) \begin{equation}\label{av_*}
|\langle x_{i_1}^{k_1}\ldots x_{i_l}^{k_l}\rangle|\le
\langle x_{i_1}^{k_1}\ldots x_{i_l}^{k_l}\rangle_*.\end{equation}
\end{lemma}
\begin{proof}
The first part of the lemma is well-known Wick's theorem, which can be easily proven using
the integration by parts.

To prove the second part set
\begin{equation}\label{M}
M=-\triangle+\gamma/W^2=(2+\gamma/W^2)I-\tilde{M},\quad
M_*=-\triangle+\Re \gamma/W^2=(2+\Re \gamma/W^2)I-\tilde{M},
\end{equation}
where $\tilde{M}=\Delta+2I$.
Then
\[
\langle x_{i} x_{j} \rangle=(M^{-1})_{ij},\quad \langle x_{i} x_{j} \rangle_*=(M^{-1}_*)_{ij}.
\]
Besides, since all entries of $\tilde{M}$ are positive and $\Re \gamma>0$,
\begin{multline*}
\big|(M^{-1})_{ij}\big|=\Big|\sum\limits_{k=0}^\infty\dfrac{(\tilde{M}^k)_{ij}}{(2+\gamma/W^2)^{k+1}}\Big|\\
\le
\sum\limits_{k=0}^\infty\dfrac{(\tilde{M}^k)_{ij}}{|2+\gamma/W^2|^{k+1}}
\le
\sum\limits_{k=0}^\infty\dfrac{(\tilde{M}^k)_{ij}}{(2+\Re \gamma/W^2)^{k+1}}=(M_*^{-1})_{ij}.
\end{multline*}
This and (\ref{Wick}) yield (ii).
\end{proof}
To leave only the quadratic form of (\ref{f_exp}) in the exponent in (\ref{F_0_1}), we have to prove
(\ref{in_sk}).
This can be done using three ideas: (1) we can replace $\langle\ldots\rangle_0$ by $\langle\ldots\rangle$
with an error which we can control; (2) using (\ref{av_*}) we can estimate the averaging
of some function over the complex measure by the averaging of the ``changed'' function (which means that we
replace all coefficients in the Taylor expansion of the function by its absolute values) over the positive
measure (we take $\Re c_{\pm}$ instead of $c_{\pm}$); (3) using Wick's theorem (\ref{Wick}) we can prove
(\ref{in_sk}) for the positive measure (see Lemma \ref{l:s_Wick}).

 Define
\begin{equation}\label{E_phi}
E_n[g]:=\exp\Big\{-\sum\limits_{j=-n}^ng(x_j/W)\Big\}
\end{equation}
for any function $g:\mathbb{R}\to \mathbb{C}$.

Then (\ref{in_sk}) can be rewritten in the form
\begin{lemma}\label{l:in}
For $E_n$ of (\ref{E_phi}) we have
\begin{equation}\label{in}
\big|\big\langle E_n[\varphi_\pm]\big\rangle_0-1\big|=o(1),
\end{equation}
where $\varphi_{\pm}$ are defined in (\ref{f_exp}).
\end{lemma}
The key point in the proof of Lemma \ref{l:in} is
\begin{lemma}\label{l:s_Wick}
Let $g$ be a polynomial of degree $q$ with real coefficients starting from the third power, i.e.
$g(x)=\sum\limits_{j=3}^q c_jx^j$, $c_j\in \mathbb{R}$. Then we have
\begin{equation}\label{in_*}
\big|\big\langle E_n[g]\big\rangle_{0,*}-1\big|=o(1),\quad n\to\infty.
\end{equation}
\end{lemma}
\begin{proof} \textbf{The lower bound.}

Since $e^x-1\ge x$, we have
\[
\big\langle E_n[g]\big\rangle_{0,*}-1\ge \big\langle \sum\limits_{j=-n}^ng(x_j/W)\big\rangle_{0,*}=\big\langle
\sum\limits_{j=-n}^ng(x_j/W)\big\rangle_{*}+o(1),
\]
where we use the third assertion of Lemma \ref{l:okr} in the last equality.
Using Wick's theorem (\ref{Wick}) and $(M_*^{-1})_{ii}=CW$ (see the assertion (1) of Lemma \ref{l:okr}),
we can write
\[
\big\langle
(x_j/W)^{2l}\big\rangle_{*}=O(W^{-l}),
\]
and hence
\[
\big\langle
\sum\limits_{j=-n}^ng(x_j/W)\big\rangle_{*}= O((2n+1)/W^2)=o(1).
\]

 \textbf{The upper bound.}

Let us prove that
\begin{equation}\label{ineqv}
\big\langle E_n[g]\big\rangle_{0,*}-1\le \varepsilon_{1,n}
\big\langle E_n[g]\big\rangle_{0,*},
\end{equation}
which implies
\[
\big\langle E_n[g]\big\rangle_{0,*}-1\le 2\varepsilon_{1,n},
\]
where $\varepsilon_{1,n}=o(1)$, as $n\to\infty$.

\textbf{Step 1. Replacing $\langle\ldots\rangle_{0,*}$ with $\langle\ldots\rangle_{*}$}

Note that if we choose $s_\kappa>3$ such that (recall that $\delta=W^{-\kappa}$, $\kappa<\theta/8$)
\begin{equation}\label{s}
W^{-\kappa s_\kappa}\le W^{-2},
\end{equation}
then for any $p>s_\kappa/2$ and for $x_j\in (-\delta W,\delta W)$
\[
\sum\limits_{j=-n}^n(x_j/W)^{2p}<N/W^2=o(1),
\]
and thus if we replace $g(x)$ by $g(x)+Cx^{2p}$ with any $C$, then $E_n[g]$ will be changed by $E_n[g](1+o(1))$.
Since it is easy to see that we can choose $C$
such that $c_0x^2/2+g(x)+Cx^{2p}$ has only one minimum $x=0$ in $\mathbb{R}$, without loss of the
generality we can assume that
$c_0 x^2/2+g(x)\ge c_0x^2/4$. Moreover, $c_0 x^2/2+g(x)\le c_0x^2$ for $x\in (-\delta,\delta)$.
This and assertions (1), (2) of Lemma \ref{l:okr} give
\begin{equation*}
\dfrac{\intd_{\max |x_i|>\delta W}E_n[g]\mu_{c_0}(x) dx}{\intd_{\max |x_i|\le\delta W}E_n[g]\mu_{c_0}(x)
dx}\le \dfrac{\intd_{\max |x_i|>\delta W}\mu_{c_0/2}(x) dx}{\intd_{\max |x_i|\le\delta W}\mu_{2c_0}(x)
dx}\le e^{Cn/W-W\delta^2}=o(1),
\end{equation*}
because $\delta=W^{-\kappa}$ with $\kappa<\theta/8$.
Thus,
\begin{equation}\label{okr*}
\big\langle E_n[g]\big\rangle_{0,*}=\big\langle E_n[g]\big\rangle_{*}+o(1).
\end{equation}

\textbf{Step 2. Application of Wick's theorem (Lemma \ref{l:Wick+} (i))}

Since for $x\in \mathbb{R}$
\[
e^x\le 1+xe^x,
\]
we can write using Wick's theorem (\ref{Wick})
\begin{align*}
&\big\langle E_n[g]\big\rangle_*-1\le \sum\limits_{i_1}
\big\langle g(x_{i_1}/W)\cdot E_n[g]\big\rangle_*=\sum\limits_{i_1}\sum\limits_{l=3}^q
\Big\langle \dfrac{c_l x_{i_1}^l}{W^l}\cdot E_n[g]\Big\rangle_*\\
&\le
\sum\limits_{i_1}\sum\limits_{l=3}^q\dfrac{(l-1)|c_l|\langle x_{i_1}^2\rangle_*}{W^2}
\Big|\Big\langle\dfrac{x_{i_1}^{l-2}}
{W^{l-2}}\cdot E_n[g]\Big\rangle_*\Big|\\
&+
\sum\limits_{i_1,i_2}\sum\limits_{l=3}^q\dfrac{|c_l|\langle x_{i_1}x_{i_2}\rangle_*}{W^2}
\Big|\Big\langle\dfrac{x_{i_1}^{l-1}}{W^{l-1}}\cdot g' \left(\frac{x_{i_2}}{W}\right)\cdot
E_n[g]\Big\rangle_*\Big|
\end{align*}
\begin{align*}
&\le
\sum\limits_{i_1}\sum\limits_{l=4}^q\dfrac{(l-1)(l-3)|c_l|\langle x_{i_1}^2\rangle_*^2}{W^4}
\Big|\Big\langle\dfrac{x_{i_1}^{l-4}}
{W^{l-4}}\cdot E_n[g]\Big\rangle_*\Big|\\
&+\sum\limits_{i_1,i_2}\sum\limits_{l=3}^q\dfrac{(2l-3)|c_l|\langle x_{i_1}x_{i_2}\rangle_*\langle x_{i_1}^2
\rangle_*}{W^4}
\Big|\Big\langle\dfrac{x_{i_1}^{l-3}}{W^{l-3}}\cdot g' \left(\frac{x_{i_2}}{W}\right)\cdot
E_n[g]\Big\rangle_*\Big|\\
&+\sum\limits_{i_1,i_2}\sum\limits_{l=3}^q\dfrac{|c_l|\langle x_{i_1}x_{i_2}\rangle_*^2}
{W^4}
\Big|\Big\langle\dfrac{x_{i_1}^{l-2}}{W^{l-2}}\cdot g'' \left(\frac{x_{i_2}}{W}\right)\cdot
E_n[g]\Big\rangle_*\Big|\\
&+\sum\limits_{i_1,i_2,i_3}\sum\limits_{l=3}^q\dfrac{|c_l|\langle x_{i_1}x_{i_2}\rangle_*\langle x_{i_1}x_{i_3}\rangle_*}
{W^4}
\Big|\Big\langle\dfrac{x_{i_1}^{l-2}}{W^{l-2}}\cdot g' \left(\frac{x_{i_2}}{W}\right)g'\left(\frac{x_{i_3}}{W}\right)\cdot
E_n[g]\Big\rangle_*\Big|=\ldots
\end{align*}
For every term $\langle x_{i_1}^{m_1}\ldots x_{i_k}^{m_k}E_n[g]\rangle_*$, we take
$x_{i_l}$ with
the smallest index $l$ and find its pair according to (\ref{Wick}).
We repeat this procedure until we get $\langle E_n[g]\rangle_*$ or until the number
of steps becomes bigger than $s_\kappa$, where $s_\kappa$ is defined in (\ref{s}). All terms have the form
\[
\sum\limits_{i_1,\ldots,i_{p+l}}G(x_{i_1},\ldots,x_{i_{p+l}})
\left|\bigg\langle\dfrac{x^{\alpha_{p+1}}_{i_{p+1}}x_{i_{p+2}}^{\alpha_{p+2}}\ldots x^{\alpha_{p+l}}_{i_{p+l}}}{W^{\alpha}}
E_n[g]\bigg\rangle_*\right|,
\]
where $\alpha_{p+1},\ldots,\alpha_{p+l}\in \mathbb{N}$ are bounded by some absolute constant
(since in any case we make a finite number of steps), $\alpha=\alpha_{p+1}+\ldots+\alpha_{p+l}$. Here
$G(x_{i_1},\ldots,i_{p+l})$ is the product of the
expectations of some pairing
$x_{i_1}^{k_1}x_{i_2}^{k_2}\ldots x_{i_{p+l}}^{k_{p+l}}$
with $k_j\ge 3$, $j=1,\ldots,p$ and $k_j\ge 1$, $j=p+1,\ldots,p+l$ (all $\{k_j\}$ are bounded by some
absolute constant) divided by $W^{k_1+\ldots+k_{p+l}}$,
with some bounded positive coefficient.

We can visualize these pairings as connected multigraphs (i.e. graphs
which may contain multiple edges and loops) with vertices $i_1,\ldots,i_{p+l}$,
where $p+l\le s_\kappa$. The degree of $i_j$ is at least $3$ for $j\le p$ and
is at least $1$ for $j=p+1,\ldots,p+l$. The multigraphs are connected,
since one can proceed to a different connected component only if we
obtained $\langle E_n[g]\rangle_*$ before.

Let $H$ be one of such multigraphs. Any $\langle x_i x_j\rangle_*$ gives $(M^{-1}_*)_{ij}$.
Thus, any loop gives a factor $(M_*^{-1})_{ii}=CW(1+o(1))$ (see the assertion (1) of Lemma \ref{l:okr}).
Moreover, according to the Cauchy-Schwarz inequality, we
have
\begin{equation*}
(M_*^{-1})_{ij}\le (M_*^{-1})_{ii}^{1/2}(M_*^{-1})_{jj}^{1/2}.
\end{equation*}
Hence, we can remove the edge $(j_1,j_2)$ from any cycle $(j_1,j_2,\ldots,j_r,j_1)$ ($r\ne1$)
and replace it with two semiloops $(j_1,j_1)$, $(j_2,j_2)$ (a ``semiloop''
is a loop counted with the coefficient $1/2$, i.e. the contribution
of a semiloop is $|(M_*^{-1})_{ii}|^{1/2}$ instead
of $(M_*^{-1})_{ii}$; one semiloop
adds 1 to the degree of the vertex, and two semiloops add up
to one loop.) In this way we transform the multigraph $H$ to a tree $H_0$ with some loops and semiloops
(the degree of each vertex is still the same as in $H$).

Since we make a finite number of steps, there is only a finite number of graphs $H$ such
that corresponding graphs $H_0$ are equal to each other. Hence, we can
consider the sum over $H_0$ instead of $H$. Let $G_0(x_{i_1},\ldots,x_{i_{p+l}})$
be the function, which corresponds to the new graph~$H_0$.

Note that, according to (\ref{okr*}),
\begin{align*}
&\Big|\Big\langle\dfrac{x^{\alpha_{p+1}}_{i_{p+1}}x_{i_{p+2}}^{\alpha_{p+2}}\ldots x^{\alpha_{p+l}}_{i_{p+l}}}
{W^{\alpha}}
E_n[g]\Big\rangle_*\Big|\le \Big|\Big\langle\dfrac{x^{\alpha_{p+1}}_{i_{p+1}}x_{i_{p+2}}^{\alpha_{p+2}}\ldots x^{\alpha_{p+l}}_{i_{p+l}}}{W^{\alpha}}
E_n[g]\Big\rangle_{0,*}+o(1)\Big|\\
&\le \dfrac{1}{W^{\kappa \alpha}}
\big\langle E_n[g]\big\rangle_{0,*}+o(1)= \dfrac{1}{W^{\kappa \alpha}}
\big\langle E_n[g]\big\rangle_{*}+o(1).
\end{align*}
Therefore, we are left to prove
\begin{lemma}\label{l:subl}
Let $H_0$ be a tree with loops and
semiloops, whose vertices $i_1,\ldots, i_{p+l}$ admit the following condition: the degree of each
vertex $i_j$ is at least $3$ for $j\le p$ and is at least $1$ for $j=p+1,\ldots,p+l$.
Denote by $m$ the sum of degrees of all vertices and
let also $G_0(x_{i_1},\ldots,x_{i_{k}})$
be the function of the pairing, which corresponds to $H_0$. Then we have
\begin{equation}\label{bound}
\sum\limits_{i_1, \ldots, i_{p+l}}G_0(x_{i_1},\ldots,x_{i_{p+l}})\le N/W^{m/2-(p+l)+1}.
\end{equation}
Moreover,
\begin{equation}\label{sum}
\sum\limits_{i_1, \ldots, i_{p+l}}W^{-\kappa\alpha}G_0(x_{i_1},\ldots,x_{i_{p+l}})=o(1).
\end{equation}
\end{lemma}
\begin{proof}
Since $(1,\ldots,1)$ is an eigenvector for $M_*$ of (\ref{M}) with eigenvalue $W^2/\Re \gamma$, we have
\begin{equation}\label{e_gr}
\sum\limits_{j} (M^{-1}_*)_{ij}=\dfrac{W^2}{\Re \gamma},\quad i=-n,\ldots,n.
\end{equation}
Let us consider the sum over $i_1,\ldots,i_{p+l}$. Any loop or semiloop gives $W$ or $W^{1/2}$
respectively. Thus, since the
tree has $p+l-1$ edges, all loops and semiloops give the contribution $W^{m/2-(p+l)+1}$.
Using (\ref{e_gr}), we obtain that
the contribution of the tree edges is $N\cdot W^{2(p+l-1)}$. Therefore, since any $x_{i_k}$
has also the coefficient $W^{-1}$, we get that the sum over $i_1,\ldots,i_{p+l}$ is bounded by
$N/W^{m/2-(p+l)+1}$.
Evidently $m$ is even, and hence $m/2-p+1$ is integer.

Now let us prove (\ref{sum}). Consider two cases

(1) \textbf{The case $l=0$.}

First consider the case when we get $\langle E_n[g]\rangle_*$ at some step. Then $l=0$ and $H_0$ is a tree
with $p$ vertices and some loops and semiloops, where the degree of each vertex is at least $3$.
Hence, $m\ge 3p$, where $m$ is the sum of all degrees, and thus $m/2-p+1\ge m/6+1>1$,
which means $m/2-p+1\ge 2$. Therefore, $N/W^{m/2-k+1}=N/W^2=o(1)$, and hence (\ref{bound})
implies (\ref{sum}).

(2) \textbf{The case $l>0$.}

Let now $l>0$. Set $k=p+l$. Using (\ref{bound}), we get that the sum over all vertices is not greater
than $N/W^{m/2-k+1}$. Since $H_0$ is a connected graph, we have
$m\ge 2(k-1)$, i.e. $m/2-k+1\ge 0$. The sum in (\ref{sum}) has a factor $W^{-\kappa \alpha}$.
In addition, $m+\alpha\ge 3k$ and $m=2s_\kappa$, since
we did $s_\kappa$ steps and thus obtained $s_\kappa$ edges. If $m/2-k+1\ge 2$, then the sum can
be bounded by
$W^{-\kappa \alpha}N/W^2=o(1)$. Hence, we are left to consider the case $0\le m/2-k+1\le 1$, i.e.
$2(k-1)\le m\le 2k$. Therefore, we get $k\ge s_\kappa$ (because $m=2s_\kappa$), thus
\[
2s_\kappa+\alpha=m+\alpha\ge 3k\ge 3 s_\kappa,
\]
which implies $\alpha>s_\kappa$. Hence, the sum in (\ref{sum}) is bounded by
$$W^{-\kappa \alpha}N/W^{m/2-k+1}\le N/W^{\kappa \alpha}=o(1),$$
which gives (\ref{sum}).
\end{proof}

Now (\ref{sum}) implies (\ref{ineqv}) with $\langle\dots\rangle_*$ and hence
with $\langle\dots\rangle_{0,*}$ (see (\ref{okr*})).

\end{proof}

\textbf{Proof of Lemma \ref{l:in}.}

We can write for $x\in (-\delta,\delta)$
\[
\widetilde{\varphi}_{\pm}(x):=\exp\{-\varphi_{\pm}(x)\}-1=\sum\limits_{l=3}^\infty \phi_lx^l,
\]
where $|\phi_l|\le (C_0)^l$. Thus,
\begin{equation}\label{in_vsp}
|\langle E_n[\varphi_{\pm}]\rangle_0-1|=
\Big|\langle \prod\limits_{j=-n}^n(1+\sum\limits_{l=3}^\infty \phi_lx^l)\rangle_0-1\Big|=\Big|
\sum\limits_{k=3}^{\infty}\Sigma_k^0\Big|,
\end{equation}
where $\Sigma_k^0$, $\Sigma_k$ are the sums of all terms
$\langle \prod_{l=1}^s(\phi_{k_l}x_{i_l}^{k_l}/W^{k_l})\rangle_0$
and $\langle\prod_{l=1}^s(\phi_{k_l}x_{i_l}^{k_l}/W^{k_l})\rangle$ respectively
with $k_1+\ldots+k_s=k$, $k_i\in \{3,\ldots,k\}$ ($\Sigma_{k,*}^0$, $\Sigma_{k,*}$ are defined by the same way
with $\langle\ldots\rangle_*$ instead of $\langle\ldots\rangle$ and with $|\phi_l|$ instead of $\phi_l$).

Also denote
\begin{equation*}
S_s^0=\sum\limits_{i_1<\ldots<i_s} \langle \prod\limits_{l=1}^s |x_{i_l}/W|^3\rangle_{0,*}.
\end{equation*}
According to (\ref{otn_mod}), we have
\[
|\langle(\phi_{k_1}x_{i_1}^{k_1}/W^{k_1})\ldots (\phi_{k_s}x_{i_s}^{k_s}/W^{k_s})\rangle_0|\le (C_0)^k\delta^{k-3s}
e^{Cn/W} \langle \prod\limits_{l=1}^s |x_{i_l}/W|^3\rangle_{0,*}.
\]
Hence, since the number of partitions of $k$ to $s$ non-zero summands is not grater than $\binom{k}{s}$, we
obtain
\begin{equation}\label{ots_sig1}
|\Sigma_k^0|\le e^{Cn/W}(C_0)^k\sum\limits_{s=1}^{k/3}\binom{k}{s}
\delta^{k-3s}S_s^0\le e^{Cn/W}(2C_0)^k\sum\limits_{s=1}^{k/3}\delta^{k-3s}S_s^0.
\end{equation}
Note now that
\begin{equation}\label{in_kub}
|x|^3\le \dfrac{p^{-1}x^2+px^4}{2},
\end{equation}
and hence, again according to (\ref{otn_mod}), we get for any $p>0$
\begin{equation}\label{s_0k}
S_s^0\le \sum\limits_{i_1<\ldots<i_s}\Big\langle \prod\limits_{l=1}^s\dfrac{p^{-1}x_{i_l}^2/W^2+px_{i_l}^4/W^4}{2}\Big\rangle_{0,*}
:=\widetilde{S}_{s}^0.
\end{equation}
Besides,
\[
1+q\cdot \dfrac{p^{-1}x^2+px^4}{2}\le (1+\dfrac{qx^2}{2p})(1+\dfrac{pqx^4}{2})\le e^{qx^2/2p}
(1+\dfrac{pqx^4}{2}),
\]
and thus, taking in account (\ref{otn_mod}), we have for any $p,q>0$ such that $q/p<c_0$ with $c_0$ of
(\ref{f*})
\begin{align*}
1+\sum\limits_{k=1}^{2n+1}q^k
\widetilde{S}_{k}^0&=\Big \langle \prod\limits_{j=-n}^n
\Big (1+q\cdot \dfrac{p^{-1}x_{j}^2/W^2+px_{j}^4/W^4}{2}\Big) \Big \rangle_{0,*}\\
&\le \Big \langle e^{q/2p\cdot\sum_jx_j^2/W^2}\prod\limits_{j=-n}^n
\Big (1+\dfrac{pq x_{j}^4/W^4}{2}\Big) \Big \rangle_{0,*}\\
&\le e^{C_{p,q}n/W}\Big \langle \prod\limits_{j=-n}^n
\Big (1+\dfrac{pq x_{j}^4/W^4}{2}\Big) \Big \rangle_{0,c_0-q/p}\le Ce^{C_{p,q}n/W},
\end{align*}
where the last inequality holds in view of Lemma \ref{l:s_Wick} ($\langle\ldots\rangle_{0,c_0-q/p}$ means (\ref{angle})
with $\gamma=c_0-q/p$). This gives
\[
\widetilde{S}_{k}^0\le e^{C_{p,q}n/W}/q^k,
\]
and we have from (\ref{s_0k}) for $k>Cn/W$ with sufficiently big $C$
\begin{equation}\label{in_s_k}
S_k^0\le e^{(C_{p,q}+C_1)n/W}q^{-k}.
\end{equation}
Take $q>(2|C_0|e)^3$. Then (\ref{ots_sig}) and (\ref{in_s_k}) yield for $k>C_1n/W$
\begin{equation}\label{ots_sig}
|\Sigma_k^0|\le e^{Cn/W}\sum\limits_{s=1}^{k/3}(2C_0)^k\delta^{k-3s}q^{-s}
\le
\dfrac{e^{Cn/W}(2C_0)^k}{q^{k/3}}\sum\limits_{s=1}^{k/3}(\delta^3q)^{k/3-s}\le 2e^{Cn/W-k}.
\end{equation}
This and (\ref{in_vsp}) imply
\begin{equation}\label{in_vsp1}
|\langle E_n[\varphi_{\pm}]\rangle_0-1|\le \Big|\sum\limits_{k=3}^{Cn/W}\Sigma_0^k
\Big |+e^{-C_1n/W}.
\end{equation}
Taking into account that the number of distributions of $k$ items into $n$ boxes is $\binom{n+k-1}{k}$
and using the assertion (3) of Lemma \ref{l:okr}, we get
\begin{align*}
|Z_\gamma|^{-1}\Big|\int\limits_{\max |x_i|>\delta W} \sum\limits_{s=1}^{k/3}\sum\limits_{k_1,\ldots,k_s}\sum_{i_1<\ldots<i_s}
\dfrac{|x_{i_1}^{k_1}\ldots x_{i_s}^{k_s}|}{W^{k}}\mu_{\gamma}(x)dx\Big|\\
\le e^{-C\delta^2W}\binom{n+k-1}{k}\le e^{2k\log (n/k)-C\delta^2W}\le e^{-C\delta^2W/4},
\end{align*}
where the second sum in the first line is over all collections $\{k_i\}_{i=1}^s$, $\sum k_i=k$, $k_i\in\{3,\ldots,k\}$.
This yields
\[
\Sigma_k=\Sigma_k^0+e^{-C\delta^2W/4},\quad \Sigma_{k,*}=\Sigma_{k,*}^0+e^{-C\delta^2W/4},\quad k\le Cn/W,
\]
and thus by
(\ref{av_*}) we have
\begin{align}\label{*}
&\Big|\sum\limits_{k=1}^{Cn/W}
\Sigma_k^0\Big|= \Big|\sum\limits_{k=1}^{Cn/W}
\Sigma_{k}\Big|+e^{-C\delta^2W/4}
\le \sum\limits_{k=1}^{Cn/W}
\Sigma_{k,*}+e^{-C\delta^2W/4}\\ \notag
&\le \sum\limits_{k=1}^{Cn/W}
\Sigma_{k,*}^0+2e^{-C\delta^2W/4}\le \langle
\prod\limits_{i=-n}^n(1+\sum\limits_{l=3}^\infty |\phi_l|x_i^l/W^l)-1\rangle_{0,*}+2e^{-C\delta^2W/4}.
\end{align}
Since $|\phi_l|\le (C_0)^l$, there exists $C$ such that
\begin{equation}\label{**}
1+\sum\limits_{l=3}^\infty |\phi_l|x^l/W^l\le e^{C(x^3/W^3+x^4/W^4)},\quad x\in(-\delta W,\delta W).
\end{equation}
This, Lemma \ref{l:s_Wick}, and (\ref{in_vsp1}) yield (\ref{in}). $\quad \quad\quad\quad\quad\quad\quad\quad
\quad\quad\quad\quad\quad\quad\quad\quad\quad\quad\quad\quad\Box$
\bigskip

Define the following partial ordering. Let $\Phi_1(x_1,\ldots,x_n)$, $\Phi_2(x_1,\ldots,x_n)$ be two analytic
functions in some ball centered at $0$, and let the coefficients of the Taylor expansion of $\Phi_2$ be
non-negative. Then we write
\begin{equation}\label{major}
\Phi_1\prec \Phi_2
\end{equation}
if the absolute value of each coefficient of the Taylor expansion of $\Phi_1$ does not exceed the corresponding
coefficient of $\Phi_2$.

 It is easy to see that
\begin{equation}\label{prop}
\Phi_3\prec \Phi_1,\quad \Phi_4\prec \Phi_2 \Rightarrow \Phi_3\Phi_4\prec \Phi_1\Phi_2.
\end{equation}

We will need
\begin{lemma}\label{l:maj}
(i) $\quad$ Let $\,|\phi_1|\le CW^{-1}$, $|\phi_2|=o(1)$ and $|\phi_k|\le C^k$ for some absolute
constant $C>0$. Then
\begin{align}\label{gen_phi}
&\langle\prod\limits_{i=-n}^n(1+\sum\limits_{l=1}^\infty |\phi_l|x_i^l/W^l)\rangle_{0,*}\le
\exp\{C|\phi_2|n/W\}.
\end{align}

(ii) $\quad$ If
\begin{equation*}
\Phi_1(s_1,\ldots,s_n)-\Phi_1(0,\ldots,0)\prec \prod\limits_{j=1}^n(1+q(s_i))-1,
\end{equation*}
where $s_i=s(\tilde{a\vphantom{b}}_i/W,\tilde{a\vphantom{b}}_{i+1}/W,\ldots,\tilde{a\vphantom{b}}_{i+k}/W,
\tilde{b}_i/W,\tilde{b}_{i+1}/W,\ldots,\tilde{b}_{i+k}/W)$ is a polynomial with $s(0,\ldots,0)=0$, $k$ is
an $n$-independent constant, and $q(s)=\sum_{j=1}^{\infty}|c_j|s^j$ with $|c_1|\le CW^{-1}$, $|c_2|=o(1)$, $|c_l|\le (C_0)^l$,
$l\ge 3$, then
\begin{equation*}
\big|\langle\Phi_1(s_1,\ldots,s_n)-\Phi_1(0,\ldots,0)\rangle_0\big|\le \langle\prod\limits_{j=1}^n(1+q(s_i^*))-1\rangle_{0,*}+e^{-Cn/W},
\end{equation*}
where $s_i^*$ is obtained from $s_i$ by replacing the coefficients of $s$ with their absolute values.
\end{lemma}
\begin{proof}
The proof of the lemma essentially repeats the proof of Lemma \ref{l:in}. Indeed,
using
\begin{align*}
|\phi_{2}|\,x^2&\le \dfrac{p^{-1}x^2+px^4}{2},\quad |\phi_2|=o(1),\\
|\phi_1||x|&\le |x|/W\le \dfrac{n^{-1}+(n/W^2)\,|x|^2}{2}\le \dfrac{p^{-1}x^2+px^4}{2}+Cn^{-1}
\end{align*}
instead of (\ref{in_kub}), we can prove (\ref{*})
for the series started from $l=1$ with $|\phi_1|\le CW^{-1}$, $|\phi_2|=o(1)$. Besides, in view of (\ref{**})
\[
1+\sum\limits_{l=1}^\infty |\phi_l|x^l/W^l\le e^{|\phi_1|\,x/W+(|\phi_2|-|\phi_1|^2/2)\,x^2/W^2+C(x^3/W^3+x^4/W^4)},
\quad x\in(-\delta W,\delta W),
\]
and hence the Cauchy-Schwarz inequality yields
\begin{align*}\notag
&\langle\prod\limits_{i=-n}^n(1+\sum\limits_{l=1}^\infty |\phi_l|x_i^l/W^l)\rangle_{0,*}\le
\Big\langle \exp\Big\{\sum\limits_{i=-n}^n2C(x_i^3/W^3+x_i^4/W^4)\Big\}\Big\rangle_{0,*}^{1/2}\\
&\times \Big\langle \exp\Big\{\sum\limits_{i=-n}^n(2|\phi_i|\,x/W+(2|\phi_2|-|\phi_1|^2)
\,x_i^2/W^2\Big\}\Big\rangle_{0,*}^{1/2}\\ \notag
&\le \exp\{C_1n|\phi_1|^2+C_2|\phi_2|n/W\}(1+o(1))\le \exp\{C|\phi_2|n/W\},
\end{align*}
where, to obtain the third line, we use Lemma \ref{l:s_Wick}
for the first factor and take the Gaussian integral
for the second factor. This proves (\ref{gen_phi}).

The second assertion of the lemma follows from the fact that
if
\begin{equation*}
\Phi_1(s_1,\ldots,s_k)\prec \Phi_2(s_1,\ldots,s_k),
\end{equation*}
then, putting $s_i=P_i(x_1,\ldots,x_k)$ for some polynomials $P_i$, we get
\begin{equation*}
\Phi_1(P_1(x_1,\ldots,x_k),\ldots, P_k(x_1,\ldots,x_k) )
\prec \Phi_2 (P_1^*(x_1,\ldots,x_k),\ldots, P_k^*(x_1,\ldots,x_k)),
\end{equation*}
where $P_i^*$ is obtained from $P_i$ by replacing the coefficients of $P_i$ with their absolute values.
In addition, there exist polynomials $S_a$, $S_b$ such that $S_a(0)=S_b(0)=0$ and
\[
s(\tilde{a\vphantom{b}}_i/W,\tilde{a\vphantom{b}}_{i+1}/W,\ldots,\tilde{a\vphantom{b}}_{i+k}/W,
\tilde{b}_i/W,\tilde{b}_{i+1}/W,\ldots,\tilde{b}_{i+k}/W)\prec \prod\limits_{j=i}^{i+k}
S_a(\tilde{a\vphantom{b}}_j/W) \prod\limits_{j=i}^{i+k} S_b(\tilde{b}_j/W).
\]
Using this two facts, one can repeat the argument of the proof of Lemma \ref{l:in}.
\end{proof}

\subsection{Integration over the unitary group $U(2)$}

The integral over the unitary group $U(2)$ can be computed using the
well-known Harish Chandra/Itsykson-Zuber formula (see e.g. \cite{Me:91}, Appendix 5)
\begin{proposition}\label{p:Its-Z}
Let $C$ be a normal $p\times p$ matrix with distinct eigenvalues $\{c_i\}_{i=1}^p$ and
$D=\hbox{diag}\{d_1,\ldots,d_p\}$, $d_i\in \mathbb{R}$. Then
\begin{equation*}
\int \exp\{t\Tr CU^*DU\} d\,\mu(U)=\Big(\prod\limits_{j=1}^{p-1}j!\Big)
\dfrac{\det[\exp\{tc_id_j\}]_{i,j=1}^n}{t^{p(p-1)/2}\triangle(C)\triangle(D)},
\end{equation*}
where $t$ is some constant and $\triangle(C)$, $\triangle(D)$ are the Vandermonde determinants
for the eigenvalues
$\{c_i\}_{i=1}^p$, $\{d_i\}_{i=1}^p$ of $C$ and $D$.

Moreover,
\begin{multline}\label{Its-Zub}
\int\limits_{U(p)}\int\limits_\Omega \exp\Big\{-\dfrac{t}{2}\Tr (C-U^*DU)^2\Big\} \triangle^2(D)f(D)
d\mu(U)d\,D\\
=\Big(\prod\limits_{j=1}^pj!\Big)\cdot t^{-p(p-1)/2} \int\limits_\Omega \exp\Big\{-\dfrac{t}{2}\sum\limits_{j=1}^p
(c_j-d_j)^2\Big\}\dfrac{\triangle(D)} {\triangle(C)} f(d_1,\ldots,d_p) d\,D,
\end{multline}
where $f(D)$ is any symmetric function of $\{d_j\}_{j=1}^p$ in the symmetric domain $\Omega$, $d\,D=
\prod\limits_{j=1}^pd\,d_j$.
\end{proposition}
The proof of the proposition can be found in \cite{Me:91}.

Moreover, it follows from the properties of the Haar measure on the unitary group U(2) that
\begin{equation*}
\int\limits_{U(2)} \prod\limits_{l,s=1}^2 V_{ls}^{p_{ls}}\bar{V}_{ls}^{q_{ls}}\exp\Big\{\Tr CV^*DV\Big\}d\mu(V)\ne 0
\end{equation*}
only if $p_{11}-q_{11}=p_{22}-q_{22}=-(p_{12}-q_{12})=-(p_{21}-q_{21})$. Since
\begin{align*}
&|(V_j)_{12}|^2=|(V_j)_{21}|^2,\\
& |(V_j)_{11}|^2=|(V_j)_{22}|^2=1-|(V_j)_{12}|^2\\
&V_{11}\bar{V}_{12}=-V_{21}\bar{V}_{22},
\end{align*}
this means that all non-zero moments of the measure $\exp\Big\{\Tr CV^*DV\Big\}d\mu(V)$ can be expressed via
expectations of $|(V_j)_{12}|^{2s}$. In addition,
\begin{equation}\label{un_mom}
\int\limits_{U(2)}
 |V_{12}|^{2s}e^{t(\Tr CV^*DV-\Tr CD)}d\mu(V)=(-1)^s\dfrac{d^s}{dx^s}\dfrac{1-e^{-x}}{x}
 \Big|_{x=t(c_1-c_2)(d_1-d_2)}.
\end{equation}

\section{Proof of the main theorem}
In this section we will prove Theorem \ref{thm:1} applying the steepest descent method
to the integral representation (\ref{F_0_1}).

\subsection{The bound for $\Sigma_c$}
\begin{lemma}\label{l:2} Let $\Sigma_c$ be the part of the integral in (\ref{F_0_1}) over the complement
of the domain $\Omega_\delta$, which is defined in (\ref{Om_delta}). Then
\begin{equation}\label{2}
|\Sigma_c|\le C_1W^{-8n-2}(4\pi)^{N}e^{-2Nc_0} e^{-C_2W^{1-2\kappa}},
\end{equation}
where $\kappa<\theta/8$ and $c_0=\Re f(a_\pm)$.
\end{lemma}
\begin{proof}
According to (\ref{F_0_1}), we have
\begin{align*}
|\Sigma_c|&\le e^{-2Nc_0}\cdot \int\limits_{\Omega_\delta^C}\exp\Big\{-
\sum\limits_{j=-n}^n(f_*(a_j)+f_*(b_j))\Big\}\\
&\times\exp\Big\{-\frac{W^2}{2}\sum\limits_{j=-n+1}^n\Tr
(V_jA_jV_j^*-A_{j-1})^2\Big\}\\
&\times\prod\limits_{l=-n}^n(a_l-b_l)^2d\,\mu(U_{-n})\,d\overline{a\vphantom{b}}\,
d\overline{b}\, \prod\limits_{p=-n+1}^nd\mu(V_p),
\end{align*}
where $f_*$ and $c_0$ are defined in (\ref{f*}). Here we insert the absolute value inside the integral
and use that
\[
\Big|\exp\Big\{-
\frac{i}{N\rho(\lambda_0)}\sum\limits_{j=-n}^n\Tr \big(P_jU_{-n}\big)^*A_j\,
(P_jU_{-n}\big)\hat{\xi}\Big\}\Big|=1.
\]
To simplify formulas below, set
\begin{equation}\label{I_0}
I_0= W^{-8n-2}(4\pi)^{N}e^{-2Nc_0}\cdot \left|\mdet^{-1}\left(-\Delta+2c_+/W^2\right)\right|.
\end{equation}
As we will see below, $I_0$ is an order of $\Sigma$ (see Lemma \ref{l:sigma}).
Also recall that, according to Lemma \ref{l:okr}, eq. (\ref{f_eq}),
\begin{equation}\label{in_det}
e^{-C_1N/W}\le \left|\mdet^{-1}\left(-\Delta+2c_+/W^2\right)\right|\le e^{-C_2N/W},
\end{equation}
and that $W^2=N^{1+\theta}$, $\kappa<\theta/8$,
and hence $CN/W\ll W^{1-2\kappa}$.

We are going to prove that
\begin{equation}\label{in_i_0}
|\Sigma_c/I_0|\le e^{-CW^{1-2\kappa}}.
\end{equation}
Using Harish Chandra/Itzykson -- Zuber formula (\ref{Its-Zub}), we get (recall that\\
$A_j=\hbox{diag}\,\{a_j,b_j\}$, $j=-n,\ldots,n$ and $\Omega_\delta^C$ is still a symmetric domain)
\begin{align} \notag
&I_0^{-1}\cdot|\Sigma_c|\le \dfrac{2^{2n}e^{-2Nc_0}}{W^{4n}I_0}\int\limits_{\Omega_\delta^C}
\exp\Big\{-\frac{W^2}{2}\sum\limits_{j=-n+1}^n\Big(
(a_j-a_{j-1})^2+
(b_j-b_{j-1})^2\Big)\Big\}\\ \label{sig_vn}
&\times\exp\Big\{-\sum\limits_{j=-n}^n(f_*(a_j)+f_*(b_j))\Big\}\,
 |(a_{-n}-b_{-n})(a_n-b_n)|\,d\overline{a\vphantom{b}}\, d\overline{b}\\ \notag
&\le CW^{-2}(2\pi)^{-N} e^{C_1N/W}\int\limits_{W\Omega_\delta^C}\exp\Big\{-\frac{1}{2}
 \sum\limits_{j=-n+1}^n\Big(
(a_j-a_{j-1})^2+(b_j-b_{j-1})^2\Big)\Big\}\\ \notag
&\times\exp\Big\{-\sum\limits_{j=-n}^n(f_*(a_j/W)+f_*(b_j/W))\Big\}\,
 |(a_{-n}-b_{-n})(a_n-b_n)|\,d\overline{a\vphantom{b}}\, d\overline{b},
\end{align}
where $f_*$ and $c_0$ are defined in (\ref{f*}). The first line here is obtained performing
recursively the integral over $V_j$ and $A_j$ starting from $j=n$ and going backwards. At each step the
integral can be written in the form (\ref{Its-Zub}), with a suitable choice of the
function $f$. In the third line we did the change
$a_j\to a_j/W$, $b_j\to b_j/W$ and used (\ref{I_0}) -- (\ref{in_det}).

Consider $a_{-n},\ldots,a_n$ (for $b_{-n},\ldots,b_n$ we have the same).
Let us divide all configurations of $\{a_j\}$ into two parts.

(i) \textbf{First part: configurations where there is at least one local large scale fluctuation.}

This means that there exists an index $j_0$ for which $|a_{j_0}-a_{j_0-1}|\ge Wn^{-\theta/4}$,
where $\theta$ is defined in the condition of Theorem \ref{thm:1}.
Let us prove that the integral (\ref{sig_vn}) over such configuration obey (\ref{2}).

Indeed, in this case
\[
\frac{1}{2}\sum\limits_{j=-n+1}^n
(a_j-a_{j-1})^2\ge CW^2n^{-\theta/2}=Cn^{1+\theta/2}.
\]
Besides, Lemma \ref{l:min_L} yields
\begin{align}\notag
f_*(x)&\ge \alpha\, (x-a_-)^2,\quad\quad x\le a_-,\\ \label{f_in}
f_*(x)&\ge \alpha\, (x-a_+)^2,\quad\quad x\ge a_+,
\end{align}
and hence the integral in (\ref{sig_vn}) over $\prod_{q=-n}^nd a_q$ can be bounded by
\begin{align*}
&\exp\{-Cn^{1+\theta/2}\}\int\limits_{|a_{j_0}-a_{j_0-1}|\ge Wn^{-\varepsilon}}
\exp\Big\{-\sum\limits_{j=-n}^nf_*(a_j/W)\Big\}\prod_{q=-n}^nd a_q\\
&\le\exp\{-Cn^{1+\theta/2}\}\int
\exp\Big\{-\sum\limits_{j=-n}^nf_*(a_j/W)\Big\}\prod_{q=-n}^nd a_q\\
&\le (C_1W)^N\exp\{-Cn^{1+\theta/2}\}\le \exp\{-Cn^{1+\theta/2}/2\}.
\end{align*}
Here we use (\ref{f_in}) to estimate the integral in the second line (after change $a_j\to Wa_j$ the integral
over $da_j$ converges and thus can be bounded by the constant). By the same way in view of (\ref{f_in}) the
integral over $\prod_{q=-n}^nd b_q$ can be bounded by $(CW)^N$ (because $\sum(b_j-b_{j-1})^2\ge 0$). Hence,
the integral (\ref{sig_vn}) over the configuration with at least one local large scale fluctuation (in
$\{a_j\}$ or $\{b_j\}$) is bounded by (recall $W^2=N^{1+\theta}$, $0<\theta\le 1$)
\[
(C_1W)^{N}W^{-2}e^{C_2N/W}\cdot \exp\{-Cn^{1+\theta/2}/2\}\le \exp\{-Cn^{1+\theta/2}/4\},
\]
which obeys (\ref{2}). Note that the expression $|(a_n-b_n)(a_{-n}-b_{-n})|$ does not play any important
role in the bounds.

(ii) \textbf{Second part: no large local scale fluctuations.}

Let now $|a_j-a_{j-1}|\le Wn^{-\theta/4}$, $j=-n+1,\ldots,n$ (and the same is valid for $\{b_j\}$).
Without loss of generality let $a_{-n}<0$.
Let $l_1$ be the first number such that $a_{l_1}>\delta W$, where $\delta>0$ is sufficiently small.
Consider the nearest to $l_1$ indices
$p_1<l_1$ and $q_1>l_1$ such that $a_{p_1}\le 0$, $a_{q_1}\le 0$. We will call the sequence
$a_{p_1+1},\ldots,a_{q_1-1}$
a ``peak''.
Remove from the sum $\sum\limits_{j=-n+1}^n
(a_j-a_{j-1})^2$ the terms $(a_{p_1+1}-a_{p_1})^2$ and $(a_{q_1}-a_{q_1-1})^2$ (the integral becomes larger).
Then take the first number $l_2>q_1$ such that $a_{l_2}>\delta W$ and
the nearest to $l_2$ indices
$p_2<l_2$ and $q_2>l_2$ such that $a_{p_2}\le 0$, $a_{q_2}\le 0$ and again remove the terms
$(a_{p_2+1}-a_{p_2})^2$ and $(a_{q_2}-a_{q_2-1})^2$, and so on (the last peak can be from
$a_{p_j+1}$ to $a_{-n}$). Assume that we obtain $k$ of such peaks.

Consider one of them. Let it consist of $m+1$ positive numbers $a_{p_r+1},\ldots,a_{p_r+m+1}=a_{q_r-1}$.
Since $|a_j-a_{j-1}|\le W n^{-\theta/4}$, we have $m\ge n^{\theta/4}\delta$
and taking into account that $a_{p_r}<0$, we have $|a_{p_r+1}/W-a_+|>\delta$, $|a_{l_r}-a_{p_r+1}|
\ge\delta W/2$. Let $Q_{p_r,m}$ be the domain of configurations such that $a_{p_r+1},\ldots,a_{p_r+m+1}$
form the peak.

Since $a_{p_r+1},\ldots,a_{p_r+m+1}>0$, according to Lemma \ref{l:min_L}, we can write
\[
f_*(a_{p_r+s}/W)\ge \alpha\, (a_{p_r+s}/W-a_+)^2,\quad s=1,\ldots,m+1.
\]
Using the inequality in the r.h.s. of (\ref{sig_vn}) and applying Lemma \ref{l:okr} to the integral
over $d a_{p_r+1}\ldots da_{p_r+m+1}$, we get (recall that
$|a_{p_r+1}/W-a_+|>\delta$ and $|a_{l_r}-a_{p_r+1}|\ge\delta W/2$)
\begin{align}\label{peak}
&\int\limits_{Q_{p_r,s}}
\exp\Big\{-\frac{1}{2}\sum\limits_{j=p_r+2}^{p_r+m+1}
(a_j-a_{j-1})^2-\sum\limits_{j=p_r+1}^{p_r+m+1}f_*(a_j/W)\Big\}d a_{p_r+1}\ldots da_{p_r+m+1}\\ \notag
&\le \int\limits_{|a_{p_r+1}-a_{l_r}|>\delta W/2}
\exp\Big\{-\frac{1}{2}\sum\limits_{j=p_r+2}^{p_r+m+1}
(a_j-a_{j-1})^2-\sum\limits_{j=p_r+1}^{p_r+m+1}\alpha\, a_j^2/W^2\Big\}d a_{p_r+1}\ldots da_{p_r+m+1}\\
\notag&\le (2\pi)^{m/2}\cdot (2\alpha)^{-1/4}W^{1/2}
\cdot\big(\sinh\dfrac{m\sqrt{2\alpha}}{W}\big)^{-1/2}\cdot e^{-C_2\delta^2W}\\
\notag&\le (2\pi)^{m/2}\cdot C_1W\cdot
e^{-m\sqrt{2\alpha}/(2W)-C_2\delta^2W}.
\end{align}
The last inequality holds since for $m\ge 1$ and large $W$
\[
\big(1-e^{-\frac{2m\sqrt{2\alpha}}{W}}\big)^{-1/2}\le C_1 W^{1/2}.
\]
Hence, for the integrals over $k$ peaks of the length $m_1,\ldots, m_k$
we obtain the bound
\[
(2\pi)^{\sum m_i/2}\,(C_1W)^{k}\exp\{-\sqrt{2\alpha}\sum m_i/(2W)\}\exp\{-C_2\delta^2Wk\}.
\]
By the same way we can estimate the integral over $da_{q_l},\ldots,da_{p_{l+1}}$ (i.e. over $a_j$'s that
lie between two peaks )
by $(2\pi)^{s/2}(C_1W)\exp\{-s\sqrt{2\alpha}/(2W)\}$, where $s=p_{l+1}-q_l+1$.
Finally, the whole integral over $\{a_j\}$ configurations with $k$ peaks which begin at
$p_{l_1},\ldots, p_{l_k}$ and end at $q_{l_1},\ldots,q_{l_k}$ can be bounded by
\[
(2\pi)^{N/2} (C_1W)^{2k+1}\exp\{-\sqrt{2\alpha}\,N/(2W)\}\exp\{-\delta^2Wk\}.
\]
The number of $\{a_j\}$ configurations with $k$ peaks is smaller than $\binom{2n+1}{2k}$
(since the number of choices of the ``beginnings'' and ``ends'' of $k$
peaks is $\binom{2n+1}{2k}$ and not all choices are suitable). Hence, we get the bound for the
integral (\ref{sig_vn}) over all $\{a_j\}$ configurations which have at least one peak:
\begin{align*}
&(2\pi)^{N/2} \exp\{-\sqrt{2\alpha}\,N/(2W)\}\sum\limits_{k=1}^{n}
\binom{2n+1}{2k}(C_1W)^{2k+1}\exp\{-\delta^2Wk\}\\
&\le (2\pi)^{N/2} \exp\{-\sqrt{2\alpha}\,N/(2W)\}\cdot W\cdot e^{-C\delta^2W}
((1+C_1W e^{-\delta^2W})^n-1)\\
&\le (2\pi)^{N/2} \exp\{-\sqrt{2\alpha}\,N/(2W)\} e^{-C_2\delta^2W}.
\end{align*}
Moreover, for the configurations of $\{a_j\}$ without peaks
$f_*(a_j)\ge \alpha\,(a_j-a_+)^2$ for each $j$, and so, according to Lemma \ref{l:okr},
the integral over $\Omega_\delta^C$ over such configurations can be bounded by
\[
(2\pi)^{N/2} C_1W^{1/2} \exp\{-\sqrt{2\alpha}\,N/(2W)\} e^{-C_2\delta^2W}.
\]
By the same way estimating the integral over $\{b_j\}$ and substituting the bounds to (\ref{sig_vn}),
we get Lemma \ref{l:2}.
\end{proof}

\subsection{Calculation of $\Sigma$}
\begin{lemma}\label{l:sigma}
For the integral $\Sigma$ over the domain $\Omega_\delta$ (see (\ref{Om_delta})) we have
\begin{align}\label{sigma}
\Sigma&=\dfrac{e^{-2Nc_0}\rho(\lambda_0)^2(4\pi)^{N}\pi^2}{W^{8n+2}}\cdot
\dfrac{\sin \pi(\xi_1-\xi_2)}{\pi(\xi_1-\xi_2)}
\cdot\Big|\mdet^{-1}\Big(-\Delta+\frac{2c_+}{W^2}\Big)\Big|(1+o(1))\\ \notag
&=(\pi\rho(\lambda_0))^2\cdot \dfrac{\sin \pi(\xi_1-\xi_2)}{\pi(\xi_1-\xi_2)}\cdot I_0 ,\quad W\to\infty,
\end{align}
where $I_0$ is defined in (\ref{I_0}).
\end{lemma}
Note that (\ref{sigma}) together with (\ref{in_i_0}) yield
\[
|\Sigma_c|\le e^{-CW^{1-2\kappa}} |\Sigma|,
\]
which gives (\ref{F_2}).

Now using (\ref{F_2}) and (\ref{sigma}) we get Theorem \ref{thm:1}.

Thus, we are left to compute $\Sigma$. We are going to show that the leading term in $\Sigma$
is given by $\Sigma_\pm$, i.e. that the contributions of $\Sigma_+$ and $\Sigma_-$ are smaller.

\subsubsection{Calculation of $\Sigma_\pm$}
Consider the $\delta$-neighborhood of the point $(\overline{a}_+,\overline{a}_-)$ with $\overline{a}_{\pm}$ of
(\ref{a_pm}) and $\delta=W^{-\kappa}$.

Let us show that
\begin{lemma}\label{l:sig}
For the integral $\Sigma_\pm$ over the domain $\Omega_\delta^\pm$ of (\ref{Om_delta}) we have
\begin{align*}\notag
\Sigma_{\pm}=\dfrac{e^{-2Nc_0}\rho(\lambda_0)^2(4\pi)^{N}\pi^2}{W^{8n+2}}\cdot\dfrac{\sin \pi(\xi_1-\xi_2)}{\pi(\xi_1-\xi_2)}
\cdot\Big|\mdet^{-1}\Big(-\Delta+\frac{2c_+}{W^2}\Big)\Big|(1+o(1)),\quad W\to\infty.
\end{align*}
\end{lemma}
\begin{proof}
%
%

Performing the change $a_j-a_+=\tilde{a\vphantom{b}}_j/W$, $b_j-a_-=\tilde{b}_j/W$ in (\ref{F_0_1})
and using (\ref{f_exp}),
we obtain (recall that $a_\pm=\pm\pi\rho(\lambda_0)$)
\begin{align}\notag
\Sigma_{\pm}=&W^{-2N}2^{2n}e^{-2Nc_0-i\pi(\xi_1-\xi_2)}
\int\limits_{|\tilde{a\vphantom{b}}_j|,|\tilde{b}_j|\le W^{1-\kappa}} \int\limits_{U(2)^N}\mu_{c_+}(a)\mu_{c_-}(b)
\\ \label{F_0_2}
&\times e^{W^2\sum\limits_{j=-n+1}^n\Tr \left(V_j^*(L+\tilde{A}_j/W)V_{j}(L+\tilde{A}_{j-1}/W)-
(L+\tilde{A}_j/W)(L+\tilde{A}_{j-1}/W)\right)}\\ \notag
& \times e^{-\sum\limits_{k=-n}^n(\varphi_+(\tilde{a}_k/W)+\varphi_-(\tilde{b}_k/W))
-\frac{i}{N\rho(\lambda_0)}\sum\limits_{k=-n}^n\big(\Tr (P_kU_{-n})^*(L+\tilde{A}_k/W)\,
(P_kU_{-n})\hat{\xi}-\Tr L\hat{\xi}\big)}
\\ \notag
&\times\prod\limits_{l=-n}^n(a_+-a_-+(\tilde{a\vphantom{b}}_l-\tilde{b}_l)/W)^2d\mu(U_{-n})\prod\limits_{q=-n+1}^n
d\mu(V_q)\,d\overline{a\vphantom{b}}
\,d\overline{b},
\end{align}
where $L=\hbox{diag}\,\{a_+,a_-\}$, $\tilde{A}_j=\hbox{diag}\,\{\tilde{a\vphantom{b}}_j,\tilde{b}_j\}$, and $\mu_{\gamma}(a)$
is defined in (\ref{mu}).

Now we are going to integrate over $\{V_j\}$.

Denote
\begin{align}\label{F_V}
F(\overline{\vphantom{b} a},\overline{b},V)&=-\frac{i}{\rho(\lambda_0)}\sum\limits_{k=-n}^n\big(\Tr
\big(P_kU_{-n}\big)^*(L+\tilde{A}_k/W)\,
(P_kU_{-n}\big)\hat{\xi}-\Tr L\hat{\xi}\big),\\ \notag
d\,\tilde{\vphantom{A}\mu}(V,\tilde{A})&=e^{W^2\sum\limits_{j=-n+1}^n\Tr \left(V_j^*(L+\tilde{A}_j/W)V_{j}(L+\tilde{A}_{j-1}/W)-
(L+\tilde{A}_j/W)(L+\tilde{A}_{j-1}/W)\right)}\prod\limits_{q=-n+1}^n d\mu(V_q),\\ \notag
I_{\tilde{\mu}}(\tilde{A})&=\int d\,\tilde{\mu\vphantom{A}}(V,\tilde{A}).
\end{align}
According to the Itsykson-Zuber formula (see Proposition \ref{p:Its-Z})
\begin{equation}\label{int_V}
I_{\tilde{\mu}}(\tilde{A})=\,W^{-4n}\prod\limits_{q=-n+1}^n\dfrac{1-e^{-W^2(a_+-a_-+(\tilde{a\vphantom{b}}_q-\tilde{b}_q)/W)
(a_+-a_-+(\tilde{a\vphantom{b}}_{q-1}-\tilde{b}_{q-1})/W)}}{(a_+-a_-+(\tilde{a\vphantom{b}}_q-\tilde{b}_q)/W)
(a_+-a_-+(\tilde{a\vphantom{b}}_{q-1}-\tilde{b}_{q-1})/W)}.
\end{equation}
We want to integrate the r.h.s. of (\ref{F_0_2}) over $d\tilde{\mu\vphantom{A}}(V,\tilde{A})$. To this end,
we expand \\$\exp\big\{F(\overline{\vphantom{b} a},\overline{b},V)\big\}$ into a series in $|(V_j)_{12}|^2$
(note that $|(V_j)_{12}|^2=|(V_j)_{21}|^2$,\\ $|(V_j)_{11}|^2=|(V_j)_{22}|^2=1-|(V_j)_{12}|^2$).
Formula (\ref{un_mom}) implies
\begin{align}
&\int |(V_j)_{12}|^{2s} d\tilde{\mu\vphantom{A}}(V,\tilde{A})\\ \notag
=&\,W^{-4n}\prod\limits_{q\ne j}\dfrac{1-e^{-W^2(a_+-a_-+(\tilde{a\vphantom{b}}_q-\tilde{b}_q)/W)
(a_+-a_-+(\tilde{a\vphantom{b}}_{q-1}-\tilde{b}_{q-1})/W)}}{(a_+-a_-+(\tilde{a\vphantom{b}}_q-\tilde{b}_q)/W)
(a_+-a_-+(\tilde{a\vphantom{b}}_{q-1}-\tilde{b}_{q-1})/W)}\\ \notag
&\times (-1)^s\dfrac{d^s}{dx^s}\dfrac{1-e^{-x}}{x}
\bigg|_{x=W^2(a_+-a_-+(\tilde{a\vphantom{b}}_{j-1}-\tilde{b}_{j-1})/W)(a_+-a_-+(\tilde{a\vphantom{b}}_j-\tilde{b}_j)/W)}.
\end{align}

We are going to show that the leading term of the integral is given by the summands without $|(V_j)_{12}|^2$.
\begin{lemma}\label{l:un} In the notations of (\ref{F_V})
\begin{equation}\label{F_ch}
\Big|
\Big\langle \Big\langle\big(\exp\{(F(\overline{\vphantom{b} a},\overline{b},V)-F(0,0,I))/N\}-1\big)
\cdot\displaystyle\Pi_1\cdot\Pi_2\Big\rangle_{0}\Big\rangle_{\tilde{\mu}}\Big|=o(1),\quad N\to\infty,
\end{equation}
where $\Pi_1$, $\Pi_2$ are the products of the Taylor's series for $\exp\{\varphi_+(\tilde{a}_j/W)\}$ and
for \\ $\exp\{\varphi_-(\tilde{b}_j/W)\}$ and
\begin{equation}\label{ang_mu}
\langle\ldots\rangle_{\tilde{\mu}}
=I_{\tilde{\mu}}(\tilde{A})^{-1}\intd (\ldots) d\tilde{\vphantom{A}\mu}(V,\tilde{A}).
\end{equation}
\end{lemma}
\begin{proof}
Since $\hat{\xi}=\frac{\xi_1+\xi_2}{2}\, I+\frac{\xi_1-\xi_2}{2a_+}\, L$, we have
\begin{align*}
&\Tr (P_kU_{-n})^*(L+\tilde{A}_k/W)
(P_kU_{-n})\hat{\xi}-\Tr (L+\tilde{A}_k/W)\hat{\xi}\\
&=\frac{\xi_1-\xi_2}{2a_+}\,\Tr ((P_kU_{-n})^*
(L+\tilde{A}_k/W)\,
(P_kU_{-n})L-(L+\tilde{A}_k/W)L)\\
&=2a_+(\xi_2-\xi_1)\cdot |(P_kU_{-n})_{12}|^2(1+
(\tilde{a\vphantom{b}}_k-\tilde{b}_k)/(a_+-a_-)W),
\end{align*}
thus
\begin{multline}\label{razn}
F(\overline{\vphantom{b} a},\overline{b},V)-F(0,0,I)\\
=\frac{2ia_+(\xi_1-\xi_2)}{\rho(\lambda_0)}
\sum\limits_{k=-n+1}^n\left(|(P_kU_{-n})_{12}|^2-|(U_{-n})_{12}|^2\right)\cdot \Big(1+
\dfrac{\tilde{a\vphantom{b}}_k-\tilde{b}_k}{(a_+-a_-)W}\Big).
\end{multline}
We can write
\begin{multline*}
\exp\Big\{\dfrac{1}{N}\Big(F(\overline{\vphantom{b} a},\overline{b},V)-F(0,0,I)\Big)\Big\}-1\\
=\sum\limits_{p=1}^\infty
\dfrac{C^p}{p!\,N^p}\sum\limits_{k_1,\ldots,k_p} \Big\langle
\prod\limits_{j=1}^p \Big[\left(|(P_{k_j}U_{-n})_{12}|^2-|(U_{-n})_{12}|^2\right)\cdot \Big(1+
\dfrac{\tilde{a\vphantom{b}}_{k_j}-\tilde{b}_{k_j}}{(a_+-a_-)W}\Big)\Big]\Big\rangle_{\tilde{\mu}},
\end{multline*}
where $\langle\ldots\rangle_{\tilde{\mu}}$ is defined in (\ref{ang_mu}).
Hence, we have to study
\begin{equation}\label{Phi}
\Phi_{k_1,\ldots,k_p}(\overline{\vphantom{b} a},\overline{b})=\Big\langle\prod\limits_{j=1}^p
\left(|(P_{k_j}U_{-n})_{12}|^2-|(U_{-n})_{12}|^2
\right)\Big\rangle_{\widetilde{\mu}}.
\end{equation}

Let $p<Cn/W$ for some constant $C$. Introduce i.i.d $\{t_j\}$ such that the density of the distribution has
the form
\begin{equation}\label{rho_t}
\rho(t_j)=\dfrac{(a_+-a_-)^2}{2}\, t_j \exp\{-t_j^2(a_+-a_-)^2\}\cdot \mathbf{1}_{0<t_j<W/2}.
\end{equation}
Consider the unitary matrices
\begin{equation}\label{V_tild}
\widetilde{V}_j=\left(\begin{array}{cc}
\tilde{r}_je^{i\tilde{\theta}_j}& \tilde{v}_je^{i\theta_j}\\
-\tilde{v}_je^{-i\theta_j}&\tilde{r}_je^{-i\tilde{\theta}_j}
\end{array}\right),
\end{equation}
where
\begin{align*}
\tilde{v}_j&=\dfrac{t_j}{W}\cdot\Big(1+\dfrac{\tilde{a\vphantom{b}}_j-\tilde{b}_j}{W(a_+-a_-)}\Big)^{-1/2}
\Big(1+\dfrac{\tilde{a\vphantom{b}}_{j-1}-\tilde{b}_{j-1}}{W(a_+-a_-)}\Big)^{-1/2},\\
\tilde{r}_j&=(1-\tilde{v}_j^2)^{1/2},
\end{align*}
and $\theta_j,\tilde{\theta}_j\in [-\pi,\pi)$.

We need
\begin{lemma}\label{l:phi}
\begin{multline}\label{Phi_tild}
\widetilde{\Phi}_{k_1,\ldots,k_p}(\overline{\vphantom{b}a},\overline{b}):=\Big\langle\prod\limits_{j=1}^p
\Big( |(\prod\limits^{-n+1}_{l=k_j}\widetilde{V}_l\cdot U_{-n})_{12}|^2-|(U_{-n})_{12}|^2\Big)
\Big\rangle_{t_j,\theta_j,\tilde{\theta}_j}\\=
\Phi_{k_1,\ldots,k_p}(\overline{\vphantom{b} a},\overline{b})+O(e^{-cW^2}),
\end{multline}
where $\langle\ldots\rangle_{t_j,\theta_j,\tilde{\theta}_j}$ means the expectation over $\{t_j\}$ with respect
to the measure with the distribution (\ref{rho_t}) and over $\{\theta_j\}$,$\{\tilde{\theta}_j\}$ from $-\pi$ to
$\pi$.
\end{lemma}
The proof of the lemma can be found in Section 6.

Denote
\begin{equation}\label{s_j}
s_j=1-\Big(1+\dfrac{\tilde{a\vphantom{b}}_j-\tilde{b}_j}{W(a_+-a_-)}\Big)
\Big(1+\dfrac{\tilde{a\vphantom{b}}_{j-1}-\tilde{b}_{j-1}}{W(a_+-a_-)}\Big).
\end{equation}
Expanding $\widetilde{V}_j$ with respect to $s_j$ we get
\[
\widetilde{V}_j=\widetilde{V}_j(0)+\dfrac{t_j}{W}((1-s_j)^{-1/2}-1)V_j^{1}+\dfrac{t_j^2}{W^2}\sum\limits_{r=1}^\infty
V^{(r)}_js_j^r,
\]
where $\widetilde{V}_j(0)$ is a unitary matrix (and hence $\|\widetilde{V}_j(0)\|\le 1$),
\[
\widetilde{V}_j^{1}=\left(
\begin{array}{cc}
0&e^{i\theta_j}\\
-e^{-i\theta_j}&0
\end{array}\right),\quad \|\widetilde{V}^{(r)}_j\|\le C^r\quad (r=1,2,\ldots),
\]
and $\{\widetilde{V}^{(r)}_j\}$ are diagonal matrices.

Since the integrals of $e^{im\theta_j}$ equal 0 for $m\ne 0$ and $2\pi$ for $m=0$, we conclude that if we
replace the coefficients in front of $e^{i\theta_j}$ and $e^{-i\theta_j}$ with the bounds for their absolute values,
then, after the averaging with respect to $\theta_j$, the resulting coefficients in front of $s_j^k$ will
grow. Hence,
\begin{align*}
\widetilde{\Phi}_{k_1,\ldots,k_p}(\overline{\vphantom{b}a},\overline{b})-
\widetilde{\Phi}_{k_1,\ldots,k_p}(0,0)\prec 4\Big(\Big\langle
\prod \Big|1+\frac{t_j}{W}\, e^{i\theta_j} s^*_jg(s^*_j)+\frac{t_j^2}{W^2}s^*_j
g(s^*_j)\Big|^{2p}\Big\rangle_{t_j,\theta_j}-1\Big),
\end{align*}
where $g(t)=C_0/(1-Ct)$ with some $n$-independent $C,C_0$ and
\begin{equation*}
s_j^*=\dfrac{\tilde{a\vphantom{b}}_j+\tilde{b}_j+
\tilde{a\vphantom{b}}_{j-1}+\tilde{b}_{j-1}}{W(a_+-a_-)}
+\dfrac{(\tilde{a\vphantom{b}}_{j-1}+\tilde{b}_{j-1})
(\tilde{a\vphantom{b}}_j+\tilde{b}_j)}{W^2(a_+-a_-)^2}.
\end{equation*}
Moreover,
\[
\Big\langle
\dfrac{t_j^{2k}}{W^{2k}}\Big\rangle_{t_j}\le \dfrac{k!}{(a_+-a_-)^{2k}W^{2k}},
\]
and thus we conclude
\[
\Big\langle
\prod \Big|1+\frac{t_j}{W}\, e^{i\theta_j} s_j^*g(s_j^*)+\frac{t_j^2}{W^2}s_j^*
g(s_j^*)\Big|^{2p}\Big\rangle_{t_j,\theta_j}\prec \prod \Big(1+\frac{2p}{W^2}s_j^*g_1(s_j^*)+
\frac{p^2}{W^2}(s_j^*)^2
g(s_j^*)^2\Big).
\]
Set
\[
\Pi_3=\prod\limits_{j=1}^p \Big(1+
\dfrac{\tilde{a\vphantom{b}}_{k_j}-\tilde{b}_{k_j}}{(a_+-a_-)W}\Big),\quad
\Pi_{3,*}=\prod\limits_{j=1}^p\Big(1+\dfrac{\tilde{a\vphantom{b}}_{k_j}+\tilde{b}_{k_j}}{(a_+-a_-)W}\Big).
\]
Since $p\le Cn/W$, we have $2p/W^2\le W^{-1}$, $p^2/W^2=o(1)$. In addition, $\Pi_3$ has degree
$p<Cn/W$, $|\Pi_3|\le (1+\delta)^p$ and thus does not spoil the bounds (\ref{in_vsp}) -- (\ref{*}).
Thus, Lemma \ref{l:maj} yields
\begin{align*}
&\Big|\Big\langle(\widetilde{\Phi}_{k_1,\ldots,k_p}(\overline{\vphantom{b}a},\overline{b})-
\widetilde{\Phi}_{k_1,\ldots,k_p}(0,0))\cdot
\displaystyle\Pi_1\cdot\Pi_2\cdot \Pi_3\Big\rangle_0\Big|\\
&\le 4\Big\langle\Big(\prod \Big(1+\frac{2p}{W^2}s_jg(s_j)+\frac{p^2}{W^2}s_j^2
g(s_j)^2\Big)-1\Big)
\cdot\displaystyle\Pi_{1,*}\cdot\Pi_{2,*}\cdot \Pi_{3,*}\Big\rangle_{0,*}+e^{-Cn/W}\\
&\le 4(1+\delta)^p\Big\langle\Big(\exp\Big\{\sum\limits_{i=-n}^n\Big(\frac{Cp}{W^2}
\cdot\frac{\tilde{a}_i+\tilde{b}_i}{W}+
\frac{p^2c}{W^2}\cdot\frac{\tilde{a}_i^2+\tilde{b}_i^2}{W^2}\Big)\Big\} -1\Big)
\cdot\displaystyle\Pi_{1,*}\cdot\Pi_{2,*}\Big\rangle_{0,*}+e^{-Cn/W}\\
&\le 4e^{\delta p}\Big\langle\Big(\exp\Big\{\sum\limits_{i=-n}^n\Big(\frac{Cp}{W^2}\cdot
\frac{\tilde{a}_i+\tilde{b}_i}{W}+
\frac{p^2c}{W^2}\cdot\frac{\tilde{a}_i^2+\tilde{b}_i^2}{W^2}\Big)\Big\} -1\Big)^2
\Big\rangle_{0,*}^{1/2}
\cdot\Big\langle\displaystyle\Pi_{1,*}^2\cdot\Pi_{2,*}^2\Big\rangle_{0,*}^{1/2}+e^{-Cn/W},
\end{align*}
where $\Pi_1$, $\Pi_2$ are the products of the Taylor's series for $\exp\{\varphi_+(\tilde{a}_j/W)\}$ and
for \\ $\exp\{\varphi_-(\tilde{b}_j/W)\}$, and $\Pi_{1,*}$, $\Pi_{2,*}$ are obtained form $\Pi_1$, $\Pi_2$ by
changing the coefficients to their absolute values.

We proved earlier (see Lemma \ref{l:s_Wick}) that the second factor is $1+o(1)$. Moreover, taking
the Gaussian integral of the first
factor (similarly to the proof of Lemma \ref{l:maj} (i)), we obtain
\begin{multline*}
\Big|\Big\langle(\widetilde{\Phi}_{k_1,\ldots,k_p}(\overline{\vphantom{b}a},\overline{b})-
\widetilde{\Phi}_{k_1,\ldots,k_p}(0,0))
\cdot\displaystyle\Pi_1\cdot\Pi_2\cdot \Pi_3\Big\rangle_0\Big|\\
\le 4e^{\delta p}
\Big(\exp\Big\{\frac{cp^2n}{W^3}\Big\}-1\Big)\le 4e^{\delta p}\Big(\exp\Big\{\frac{cpn^2}{W^4}\Big\}-1\Big),
\end{multline*}
and thus, since $p<Cn/W$,
\begin{multline}\label{phi_bound}
\sum\limits_{p=1}^{Cn/W}\dfrac{(C_1)^p}{p!N^p}\sum\limits_{k_1,\ldots,k_p} \Big|
\Big\langle(\widetilde{\Phi}_{k_1,\ldots,k_p}
(\overline{\vphantom{b}a},\overline{b})-\widetilde{\Phi}_{k_1,\ldots,k_p}(0,0))
\cdot\displaystyle\Pi_1\cdot\Pi_2\cdot\Pi_3\Big\rangle_0\Big|\\
\le \exp\{e^{C_2n^2/W^4+C_1\delta}\}-e^{C_1\delta}=o(1).
\end{multline}
If $p\gg n/W$, then $1/\sqrt{p!}\ll e^{-Cn/W}$, and hence we can replace $\langle\ldots\rangle_0$ with
$\langle\ldots\rangle_{0,*}$
(see Lemma~\ref{l:okr}) and then take the absolute value under the integral and get the bound\\
\[e^{C_1n/W}([\sqrt{Cn/W}]!)^{-1}\sum\limits_{p=CN/W}^\infty (C_2)^p/\sqrt{p!}=o(1).\]

Let us prove now that
\[
\widetilde{\Phi}_{k_1,\ldots,k_p}(0,0)=\Big\langle\prod\limits_{j=1}^p
\Big( |(\widetilde{P}_{k_j}(0)U_{-n})_{12}|^2-|(U_{-n})_{12}|^2\Big)
\Big\rangle_{t_j,\theta_j,\tilde{\theta}_j}=o(1),
\]
where
\[
\widetilde{P}_{k_j}(0)=\prod\limits^{-n+1}_{l=k_j}\widetilde{V}_l(0).
\]
To this end, we write
\begin{align*}
&\Big\langle\prod\limits_{j=1}^p
\Big| |(\widetilde{P}_{k_j}(0)U_{-n})_{12}|^2-|(U_{-n})_{12}|^2\Big|
\Big\rangle_{t_j,\theta_j,\tilde{\theta}_j}\le \Big\langle
\Big| |(\widetilde{P}_{k_1}(0)U_{-n})_{12}|^2-|(U_{-n})_{12}|^2\Big|
\Big\rangle_{t_j,\theta_j,\tilde{\theta}_j}\\
&\le \Big\langle
\Big| |(\widetilde{V}_{k_1}(0))_{12}(\widetilde{P}_{k_1-1}(0)U_{-n})_{22}+
(\widetilde{V}_{k_1}(0))_{11}(\widetilde{P}_{k_1-1}(0)U_{-n})_{12}|^2-|(U_{-n})_{12}|^2\Big|
\Big\rangle_{t_j,\theta_j,\tilde{\theta}_j}\\
&=\Big\langle|(\widetilde{P}_{k_1-1}(0)U_{-n})_{22}|^2\Big\rangle_{t_j,\theta_j,\tilde{\theta}_j}\cdot
\Big\langle|(\widetilde{V}_{k_1}(0))_{12}|^2\Big\rangle_{t_j,\theta_j,\tilde{\theta}_j}
\\&+\Big\langle
\Big| |(\widetilde{P}_{k_1-1}(0)U_{-n})_{12}|^2-|(U_{-n})_{12}|^2\Big|
\Big\rangle_{t_j,\theta_j,\tilde{\theta}_j}\\&\le
\dfrac{C}{W^2}+\Big\langle
\Big| |(\widetilde{P}_{k_1-1}(0)U_{-n})_{12}|^2-|(U_{-n})_{12}|^2\Big|
\Big\rangle_{t_j,\theta_j,\tilde{\theta}_j}\le \ldots\le \dfrac{CN}{W^2}=o(1).
\end{align*}
This yields
\begin{multline*}
\sum\limits_{p=1}^{Cn/W}\dfrac{(C_1)^p}{p!N^p}\sum\limits_{k_1,\ldots,k_p} \Big|\Big\langle
\widetilde{\Phi}_{k_1,\ldots,k_p}(0,0)\cdot\displaystyle\Pi_1\cdot\Pi_2\cdot\Pi_3\Big\rangle_0\Big|\\
\le \dfrac{CN}{W^2}\sum\limits_{p=1}^{Cn/W}\dfrac{(C_1)^p(1+\delta)^p}{p!}\le C_1N/W^2=o(1),
\end{multline*}
which together with (\ref{phi_bound}) completes the proof of Lemma \ref{l:un}.
\end{proof}

Thus, we can change $F(\overline{\vphantom{b} a},\overline{b},V)$ to $F(0,0,I)$ in (\ref{F_0_2}), and
then integrate over
$\tilde{\mu}$, according to
(\ref{int_V}). We obtain
\begin{align}\notag
\Sigma_{\pm}=&W^{-8n-2}2^{2n}e^{-2Nc_0}
\int\limits_{U(2)}\,\,\,\int\limits_{|\widetilde{\vphantom{b}a}_j|, |\tilde{b}_j|\le W^{1-\kappa}} \mu_{c_+}(a)\,\mu_{c_-}(b)\\ \label{F_0_3}
&\times
\exp\Big\{-\sum\limits_{j=-n}^n\varphi_+(\tilde{a\vphantom{b}}_j/W)-\sum\limits_{j=-n}^n\varphi_-(\tilde{b}_j/W)\Big\}\\ \notag
&\times e^{-\frac{i}{\rho(\lambda_0)}\,\Tr U_{-n}^*LU_{-n}\hat{\xi}}\,\,
(a_+-a_-+(\tilde{a\vphantom{b}}_{-n}-\tilde{b}_{-n})/W)\\ \notag
&\times(a_+-a_-+(\tilde{a\vphantom{b}}_n-\tilde{b}_n)/W)d\,\mu(U_{-n})\prod\limits_{q=-n}^nd \tilde{a\vphantom{b}}_q\,
d \tilde{b}_q (1+o(1))
\end{align}
Integrating over $U_{-n}$ by the Itsykson-Zuber formula (see Proposition \ref{p:Its-Z}) and
using Lemma \ref{l:in}, we get finally
\begin{align}\notag
\Sigma_{\pm}=&\dfrac{W^{-8n-2}2^{2n}e^{-2Nc_0}(e^{i\pi(\xi_1-\xi_2)}-e^{i\pi(\xi_2-\xi_1)})}{2i\pi(\xi_1-\xi_2)}\int
\limits_{|\widetilde{\vphantom{b}a}_j|, |\tilde{b}_j|\le W^{1-\kappa}} \prod\limits_{q=-n}^nd \tilde{a\vphantom{b}}_q\,
d \tilde{b}_q \cdot
 \mu_{c_+}(a)\,\mu_{c_-}(b)\\ \label{F_0_4}
&\times
(a_+-a_-+(\tilde{a\vphantom{b}}_{-n}-\tilde{b}_{-n})/W)(a_+-a_-+(\tilde{a\vphantom{b}}_n-\tilde{b}_n)/W)(1+o(1))\\
\notag
&=\dfrac{2\pi^2 e^{-2Nc_0}\rho(\lambda_0)^2(4\pi)^{N} \sin(\pi(\xi_1-\xi_2))}{W^{8n+2}\cdot \pi(\xi_1-\xi_2)}\,
 \Big|\mdet^{-1}\Big(-\Delta+\frac{2c_+}{W^2}\Big)\Big| (1+o(1)).
\end{align}
\end{proof}

\subsubsection{$\Sigma_+$ and $\Sigma_-$.}
In this section we prove that the integrals $\Sigma_+$ and $\Sigma_-$ over $\Omega_\delta^+$ and
$\Omega_\delta^-$ have smaller orders than $\Sigma_\pm$.

Similarly to (\ref{F_0_2}) we get
\begin{align}\notag
\Sigma_{+}=&W^{-8n-4}2^{2n}e^{-2Nf(a_+)-i\pi(\xi_1+\xi_2)}
\int\limits_{|\tilde{a\vphantom{b}}_j|,|\tilde{b}_j|\le W^{1-\kappa}} \int\limits_{U(2)}\mu_{c_+}(a)\mu_{c_+}(b)
\\ \label{F_0_2_+}
&\times e^{\sum\limits_{j=-n+1}^n\Tr \left(V_j^*\tilde{A}_jV_{j}\tilde{A}_{j-1}-
\tilde{A}_j\tilde{A}_{j-1}\right)}\\ \notag
& \times e^{-\sum\limits_{k=-n}^n(\varphi_+(\tilde{a}_k/W)+\varphi_+(\tilde{b}_k/W))
-\frac{i}{N\rho(\lambda_0)}\sum\limits_{k=-n}^n\Tr (P_kU_{-n})^*(\tilde{A}_k/W)\,
(P_kU_{-n})\hat{\xi}}
\\ \notag
&\times\prod\limits_{l=-n}^n(\tilde{a\vphantom{b}}_l-\tilde{b}_l)^2d\mu(U_{-n})\prod\limits_{q=-n+1}^n
d\mu(V_q)\,d\overline{a\vphantom{b}}
\,d\overline{b}(1+o(1)).
\end{align}
By the same argument as for $\Sigma_{\pm}$ we get
\begin{align}\notag
\Sigma_{+}=&2^{2n}W^{-8n-4}e^{-2Nf(a_+)-i\pi(\xi_1+\xi_2)} \int\limits_{|\widetilde{\vphantom{b}a}_j|, |\tilde{b}_j|\le W^{1-\kappa}} \prod\limits_{q=-n}^nd \tilde{a\vphantom{b}}_q\,
d \tilde{b}_q\,\,\\ \label{sig2_++}
&\times \mu_{c_+}(a)\mu_{c_+}(b)\,
(\tilde{a\vphantom{b}}_{-n}-\tilde{b}_{-n})\,(\tilde{a\vphantom{b}}_{n}-\tilde{b}_{n})\,(1+o(1))\\
=&\notag 2^{2n}W^{-4(2n+1)}e^{-2Nf(a_+)-i\pi(\xi_1+\xi_2)}\int_{\mathbb{R}} \prod\limits_{q=-n}^nd
\tilde{a\vphantom{b}}_q\,
d \tilde{b}_q \\
\notag
&\times \mu_{c_+}(a)\mu_{c_+}(b)\,
(\tilde{a\vphantom{b}}_{-n}-\tilde{b}_{-n})\,(\tilde{a\vphantom{b}}_{n}-\tilde{b}_{n})\,(1+o(1))\\ \notag
=&(4\pi)^{N}W^{-8n-4}e^{-2Nf(a_+)-i\pi(\xi_1+\xi_2)} D^{-1}_{-n,n}\mdet^{-1} D,
\end{align}
where
\begin{align*}
D=-\Delta+\dfrac{2c_+}{W^2}.
\end{align*}
It is easy to see (see the proof of Lemma \ref{l:okr}) that for $W^2=N^{1+\theta}$, $0<\theta\le 1$
\[
|D^{-1}_{-n,n}|=1/|\mdet D|\le CW.
\]
%
Hence, since $\Re f(a_+)=c_0$, we get
\[
|\Sigma_{+}|\le C(4\pi)^{N}W^{-8n-3}e^{-2Nc_0}|\mdet^{-1} D|\le CW^{-1}|\Sigma_\pm|,
\]
and thus the order of $\Sigma_+$ is smaller than the order of $\Sigma_{\pm}$.
This completes the proof of Lemma \ref{l:sigma}.


\section{Auxiliary result}

\textbf{Proof of Lemma \ref{l:min_L}.}
Note that
\begin{align*}
f_*(a_{\pm})&=0,\quad
\dfrac{d}{dx} f_*(x)\Big|_{x=a_\pm}=0,\quad
\dfrac{d^2}{dx^2} f_*(x)\Big|_{x=a_\pm}=
2(1-\lambda_0^2/4)>0.
\end{align*}
Thus, function $f_*(x)$ attains its minimum at $a_{\pm}$ and expanding $f_*(x)$ in $x\in (a_\pm-\delta,a_\pm+\delta)$ we get
\begin{equation}
f_*(x)=(1-\lambda_0^2/4)(x-a_{\pm})^2+O(\delta^3).
\end{equation}
This yields (\ref{ineqv_ReV}). Besides, it is easy to see, that if we take
$\alpha=\frac{1}{2}(1-\lambda_0^2/4)$, then we obtain (\ref{in_left}) for some sufficiently small $\delta>0$.
$\quad \Box$
 \medskip

\textbf{Proof of Lemma \ref{l:okr}}

1) Set $-\Delta_1=-\Delta+E_{0}$, where $E_{0}$ is an $N\times N$ matrix whose elements
are zeros except $(E_0)_{-n,-n}=1$.

Define
\begin{equation}\label{T,S_n}
T_n(x)=\mdet\, (-\Delta_1+x\cdot I),\quad S_n(x)=\mdet\, (-\Delta+x\cdot I).
\end{equation}
It is easy to check that
\begin{align}\label{t_n}
T_n(x)&=(2+x)T_{n-1}(x)-T_{n-2}(x),\quad T_1(x)=1+x, \quad T_2(x)=x^2+3x+1,\\ \label{s_n}
S_n(x)&=(1+x)T_{n-1}(x)-T_{n-2}(x).
\end{align}
Solving the recurrent relation (\ref{t_n}), we get
\begin{equation}\label{rec}
T_m(x)=\dfrac{\zeta^{m+1}+\zeta^{-m}}{\zeta+1},\quad
S_m(x)=\dfrac{(\zeta^{m}-\zeta^{-m})(\zeta-1)}{\zeta+1}
\end{equation}
where
\[
\zeta=\dfrac{2+x+\sqrt{x^2+4x}}{2}.
\]
For $x=2\gamma/W^2$
\[
\zeta= 1+\sqrt{2\gamma}/W+\gamma/W^2+O(W^{-3}),\quad
W\to \infty.
\]
This and (\ref{s_n}) -- (\ref{rec}) yield
\begin{equation*}
T_m(2\gamma/W^2)=\cosh\dfrac{m\sqrt{2\gamma}}{W}(1+o(1)),\quad
S_m(2\gamma/W^2)=\dfrac{\sqrt{2\gamma}}{W}\sinh \dfrac{m\sqrt{2\gamma}}{W}(1+o(1)),
\end{equation*}
and thus (\ref{f_eq}).
Also it is easy to see that
\begin{equation*}
G_{ii}^{(m)}(\gamma)=\dfrac{T_{i-1}(2\gamma/W^2)T_{m-i}(2\gamma/W^2)}{S_{m}(2\gamma/W^2)}
\le \dfrac{C_\gamma W}{\sqrt{2\gamma}}\coth\dfrac{m\sqrt{2\gamma}}{W}(1+o(1)).
\end{equation*}
Moreover,
\begin{multline*}
G_{11}^{(m)}(\gamma)-G_{1m}^{(m)}(\gamma)=\dfrac{T_{m-1}(2\gamma/W^2)-1}{S_{m}(2\gamma/W^2)}\\=
C_\gamma W \coth\dfrac{m\sqrt{2\gamma}}{2W}(1+o(1))\le C_\gamma^1\min\{m,W\}.
\end{multline*}


2) Take $m\ge CW$ and $\alpha\in \mathbb{R}, \alpha>0$. Note that for any sufficiently small $\delta>0$ and $\varepsilon>0$
\begin{align}\notag
Z_{\alpha}^{(m)}-Z_{\delta,\alpha}^{(m)}=&\int\limits_{\max |x_i|>\delta
W}e^{-\frac{1}{2}\sum\limits_{j=2}^{m}(x_j-x_{j-1})^2-\frac{\alpha}{W^2}\sum\limits_{j=1}^{m}x_j^2}
\prod\limits_{q=1}^md x_q\\ \label{in1} &\le \sum\limits_{i=1}^m\int
e^{\frac{\varepsilon^2}{2}(x_i^2-W^2\delta^2)-\frac{1}{2}\sum\limits_{j=2}^{m}(x_j-x_{j-1})^2-
\frac{\alpha}{W^2}\sum\limits_{j=1}^{m}x_j^2} \prod\limits_{q=1}^md x_q\\ \notag
&=\sum\limits_{i=1}^m\dfrac{e^{-\varepsilon^2\delta^2W^2/2}}{\sqrt{2\pi}}\int dt\,\, e^{-t^2/2} \int \prod\limits_{q=1}^md x_q
\, e^{\varepsilon t x_i-\frac{1}{2}\sum\limits_{j=2}^{m}(x_j-x_{j-1})^2-
\frac{\alpha}{W^2}\sum\limits_{j=1}^{m}x_j^2}\\ \notag &=\dfrac{m\,e^{-\varepsilon^2\delta^2W^2/2}}{\sqrt{2\pi}} \cdot
Z^{(m)}_\alpha \cdot \sum\limits_{i=1}^m\int
e^{-t^2/2+\varepsilon^2G^{(m)}_{ii}(\alpha)t^2/2}dt,
\end{align}
where $G^{(m)}$ is defined in (\ref{G}).

Let us take $\varepsilon^2=(G^{(m)}_{ii}(\alpha))^{-1}/2$ in (\ref{in1}). Then taking into account (\ref{G_as}) and
$CW\le m\le
2n+1$, we obtain for $\alpha\in \mathbb{R}, \alpha>0$
\begin{equation}\label{in2}
\dfrac{Z_{\alpha}^{(m)}-Z_{\delta,\alpha}^{(m)}}{Z_{\alpha}^{(m)}}\le C_1\, e^{-C_2\delta^2 W}.
\end{equation}
Since $m\le 2n+1$, according to the first assertion of the lemma, we get
\begin{equation*}
\dfrac{|Z^{(m)}_{\gamma_1}|}{|Z^{(m)}_{\gamma_2}|}=(1+C/W)^m\le e^{C_1m/W}, \quad m,W\to\infty,
\end{equation*}
which gives (\ref{otn_mod}).
This and (\ref{in2}) yield for $m\ge CW$, $\gamma\in \mathbb{C}$, $\Re \gamma>0$
\begin{align*}
\dfrac{|Z_{\gamma}^{(m)}-Z_{\delta,\gamma}^{(m)}|}{|Z^{(m)}_\gamma|}\le\dfrac{Z_{\Re\gamma}^{(m)}-
Z_{\delta,\Re\gamma}^{(m)}}{Z^{(m)}_{\Re \gamma}}\cdot
 \dfrac{Z^{(m)}_{\Re \gamma}}{|Z^{(m)}_\gamma|}
\le C_1\, e^{-C_2\delta^2 W+Cn/W}\le C_1\, e^{-C_3\delta^2 W}.
\end{align*}
Since $W^2=N^{1+\theta}$, we can take $\delta=W^{-\kappa}$ with $\kappa<\theta/(1+\theta)$.

Take now any $m$. Using the assertion (1) of the lemma, we can write for any $\varepsilon>0$
\begin{align*}
&(Z_{\alpha}^{(m)})^{-1}\int\limits_{x_k-x_1>\delta
W}e^{-\frac{1}{2}\sum\limits_{j=2}^{m}(x_j-x_{j-1})^2-\frac{\alpha}{2W^2}\sum\limits_{j=1}^{m}x_j^2}
\prod\limits_{q=1}^md x_q\\
&\le (Z_{\alpha}^{(m)})^{-1}\int e^{\varepsilon(x_k-x_1-\delta W)-\frac{1}{2}\sum\limits_{j=2}^{k}(x_j-x_{j-1})^2-\frac{\alpha}{2W^2}\sum\limits_{j=1}^{k}x_j^2}
\prod\limits_{q=1}^kd x_q\\
&\times \int e^{-\frac{1}{2}\sum\limits_{j=k+2}^{m}(x_j-x_{j-1})^2-\frac{\alpha}{2W^2}
\sum\limits_{j=k+1}^{m}x_j^2}
\prod\limits_{q=k+1}^md x_q\\
&\le \dfrac{Z_{\alpha}^{(k)}Z_{\alpha}^{(m-k)}}{Z_{\alpha}^{(m)}}\cdot e^{-\varepsilon \delta W+c\varepsilon^2
(G^{(k)}_{11}-G^{(k)}_{1k})}\le
W e^{-\varepsilon \delta W+C\varepsilon^2 \min\{m,W\}}\le e^{-C_1\delta^2W}.
\end{align*}

3) It is easy to see that
\[
-\dfrac{\alpha x^2}{2}+k_i\log |x|\le -\dfrac{\alpha x^2}{4} + \dfrac{k_i}{2}\log \dfrac{2k_i}{\alpha}.
\]
Thus, using the assertions (1) -- (2) of the lemma, we obtain
\begin{align*}
&|Z^{(m)}_\gamma|^{-1}\bigg|\,\displaystyle\int\limits_{\max |x_i|>\delta W}\prod\limits_{j\in S}x_j^{k_j}
\cdot \mu^{(m)}_{\gamma}(x)
\prod\limits_{q=1}^md x_q\bigg|\\
&\le |Z^{(m)}_\gamma|^{-1}e^{\sum\limits_{i=1}^s\frac{k_i}{2}\log \frac{2k_i}{\Re \gamma}} \cdot \displaystyle\int\limits_{\max |x_i|>\delta W}
\mu^{(m)}_{\Re \gamma/2}(x)
\prod\limits_{q=1}^md x_q \\
&\le e^{C_1k\log k+C_2m/W}\dfrac{|Z_{\Re \gamma/2}^{(m)}-Z_{\delta,\Re \gamma/2}^{(m)}|}{|Z^{(m)}_{\Re \gamma/2}|}\le e^{-CW\delta^2},
\end{align*}
where the last inequality holds since $k\le Cm/W\ll W$.
$\quad \Box$
 \medskip

\textbf{Proof of Lemma \ref{l:phi}.}
Recall that all non-zero moments of measure $\tilde{\mu}$ can be expressed via
expectations of $|(V_j)_{12}|^{2s}$ (see Section 4.3). In addition, according to (\ref{int_V}),
\[
\langle|(V_j)_{12}|^{2s}\rangle_{V_j}=\dfrac{s!}{W^{2s}(a_+-a_-)^{2s}}
\Big(1+\dfrac{\tilde{a\vphantom{b}}_{j-1}-\tilde{b}_{j-1}}{W(a_+-a_-)}\Big)^{-s}
\Big(1+\dfrac{\tilde{a\vphantom{b}}_j-\tilde{b}_j}{W(a_+-a_-)}\Big)^{-s}+O(e^{-C_1W^2}).
\]
Besides,
\begin{equation*}
\int \prod\limits_{l,s=1}^2 \tilde{V}_{ls}^{p_{ls}}\bar{\tilde{V}}_{ls}^{q_{ls}}\rho(t) dtd\theta
d\tilde{\theta}\ne 0
\end{equation*}
only if $p_{11}-q_{11}=p_{22}-q_{22}$, $p_{12}-q_{12}=p_{21}-q_{21}$, and all non-zero
moments of the measure with respect to $t_j,\theta_j,\tilde{\theta}_{j}$ can be expressed via
the expectations of $|(\tilde{V}_j)_{12}|^{2s}$. Moreover,
\begin{multline*}
\langle|(\tilde{V}_j)_{12}|^{2s}\rangle_{t_j,\theta_j,\tilde{\theta}_j}\\
=\dfrac{s!}{W^{2s}(a_+-a_-)^{2s}}
\Big(1+\dfrac{\tilde{a\vphantom{b}}_{j-1}-\tilde{b}_{j-1}}{W(a_+-a_-)}\Big)^{-s}
\Big(1+\dfrac{\tilde{a\vphantom{b}}_j-\tilde{b}_j}{W(a_+-a_-)}\Big)^{-s}+O(e^{-C_2W^2}).
\end{multline*}
Hence, if $\sum p_{ls}=\sum q_{ls}$, $0\le p_{ls},q_{ls}\le 2p$, then
\[
\langle\prod\limits_{l,s=1}^2 V_{ls}^{p_{ls}}\bar{V}_{ls}^{q_{ls}}\rangle_V=\langle\prod\limits_{l,s=1}^2 \tilde{V}_{ls}^{p_{ls}}\bar{\tilde{V}}_{ls}^{q_{ls}}
\rangle_{t,\theta,\tilde{\theta}}+O(e^{-CW^2})
\]
Now let $\mathbf{E}_k$ be the averaging with respect to the product of the measures
$t_j,\theta_j,\tilde\theta_j$ for $j$ from $(-n+1)$ to $(-n+k)$ and the measures $d\mu(V_j)$ for
$j$ from $(-n+k+1)$ to $n$. Thus, if
\[\Psi_{k_1,\ldots,k_s}=\prod_{j=1}^s |(P_{k_j}U_{-n})_{12}|^{2},\]
then it suffices to estimate
\[
|(\mathbf{E}_{0}-\mathbf{E}_{2n})\{\Psi_{k_1,\ldots,k_s}\}|\le e^{-cW^2}
\]
for $s\le p$. Note that
\[
|(\mathbf{E}_{0}-\mathbf{E}_{2n})\{\Psi_{k_1,\ldots,k_s}\}|
\le\sum_{i}|(\mathbf{E}_{i-1}-\mathbf{E}_{i})\{\Psi_{k_1,\ldots,k_s}\}|\]
In each summand we write for $\gamma=i-1,i$ (we assume that all $k_j\ge (-n+i)$)
\begin{align*}
\mathbf{E}_{\gamma}\{\Psi_{k_1,\ldots,k_s}\}=&\mathbf{E}_{\gamma}\{\prod_{j=1}^s
|(P_{-n+i-1}V_{-n+i}(P_{-n+i}^*P_{k_j}U_{-n}))_{12}|^{2}\}\\=&
\mathbf{E}_{\gamma}\{\prod_{j=1}^s |\sum_{\alpha,\alpha'=1,2}
(P_{-n+i-1})_{1\alpha}(V_{-n+i})_{\alpha\alpha'}
(P_{-n+i}^*P_{k_j}U_{-n}))_{\alpha'2}|^{2}\}\\&=\sum_{l=1}^s C_{l}\mathbf{E}_{\gamma}\{|(V_i)_{12}|^{2l}\},
\end{align*}
where the coefficients $C_{l}$ are the same for $\gamma=i$ and $\gamma=i-1$ and can be bounded by $C^{p}$,
since $|(P_{-n+i-1})_{1\alpha}|\le 1$ and $|(P_{-n+i}^*P_{k_j}U_{-n}))_{\alpha'2}|\le 1$.
Moreover, since
\[|\mathbf{E}_{i}\{|(V_i)_{12}|^{2l}\}-\mathbf{E}_{i-1}\{|(V_i)_{12}|^{2l}\}|\le C^pp!e^{-CW^2},\]
we obtain
\[|(\mathbf{E}_{0}-\mathbf{E}_{2n})\{\Psi_{k_1\ldots,k_s}\}|\le nC^p_1p!e^{-CW^2}\]
Then the summation with respect to $s$ gives the bound $nC_1^{C_2 n/W}e^{-CW^2}=O(e^{-cW^2})$.
This yields Lemma \ref{l:phi}, since the expression under the expectation
in (\ref{Phi}) has the same number of elements of $V_j$ and $V_j^*$.
$\quad \Box$
 \medskip

{\bf Acknowledgements.}
 I am grateful to Thomas Spencer who drew my attention to this problem. Also I would like to thank
 both referees for their substantial efforts and useful comments which help to make the presentation
 of the paper much more clear.

\end{document}